\documentclass[reqno,10pt]{amsart}
\usepackage{amsmath,amssymb}
\usepackage[mathscr]{euscript}
\usepackage{longtable}

\usepackage{mathtools} 
\usepackage{enumerate}

\usepackage{epic}

\usepackage{color}
\definecolor{light}{gray}{.9}

\def\tpl{{| \mskip -1.5mu | \mskip -1.5mu |}}
\newcommand{\vertiii}[1]{{\tpl #1 \tpl}}

\theoremstyle{plain}
\newtheorem{theorem}{Theorem}[section]
\newtheorem{proposition}[theorem]{Proposition}
\newtheorem{lemma}[theorem]{Lemma}

\numberwithin{equation}{section}
\numberwithin{theorem}{section}

\theoremstyle{definition}

\theoremstyle{remark}
\newtheorem{remark}[theorem]{Remark}

\newcommand{\mc}[1]{{\mathcal #1}}

\newcommand{\ms}[1]{{\mathscr #1}}
\newcommand{\bb}[1]{{\mathbb #1}}

\DeclareMathOperator{\tr}{Tr}

\newcommand{\varch}{\mathop{\rm ch}\nolimits}

\newcommand{\Ad}{\mathop{\rm Ad}\nolimits}

\newcommand{\diam}{\mathop{\rm diam}\nolimits}

\newcommand{\gap}{\mathop{\rm gap}\nolimits}

\newcommand{\re}{\mathop{\rm Re}\nolimits}
\newcommand{\im}{\mathop{\rm Im}\nolimits}
\newcommand{\Ran}{\mathop{\rm Ran}\nolimits}

\newcommand{\Pfin}{\ms P}

\newcommand{\id}{{1 \mskip -5mu {\rm I}}}
\newcommand{\ind}{\mathbf{ 1}}

\renewcommand{\epsilon}{\varepsilon}
\renewcommand{\tilde}{\widetilde}

\begin{document}
\title[Interacting quantum systems]{
Perturbative criteria for the ergodicity of interacting dissipative quantum lattice systems
}

\author [L.\ Bertini]{L. Bertini}
\address{Lorenzo Bertini 
\hfill\break \indent
Dipartimento di Matematica, Universit\`a di Roma La Sapienza,
\hfill\break \indent
P.le A. Moro 5, I-00185 Roma, Italy}
\email{bertini@mat.uniroma1.it}
\author [A.\ De Sole]{A. De Sole}
\address{Alberto De Sole
\hfill\break \indent
Dipartimento di Matematica, Universit\`a di Roma La Sapienza, \& INFN,
\hfill\break \indent
P.le A. Moro 5, I-00185 Roma, Italy}
\email{alberto.desole@uniroma1.it}
\author [G.\ Posta]{G. Posta}
\address{Gustavo Posta\hfill\break \indent
  Dipartimento di Matematica, Universit\`a di Roma La Sapienza,
  \hfill\break \indent
P.le A. Moro 5, I-00185 Roma, Italy}
\email{gustavo.posta@uniroma1.it}
\author [C.\ Presilla]{C. Presilla}
\address{Carlo Presilla\hfill\break \indent
  Dipartimento di Matematica, Universit\`a di Roma La Sapienza, \& INFN,
  \hfill\break \indent
P.le A. Moro 5, I-00185 Roma, Italy}
\email{carlo.presilla@uniroma1.it}

\noindent


\begin{abstract}
  We analyze a class of quantum Feller semigroups describing the
  evolution of interacting quantum lattice systems, specified either
  as generic qudits or as fermions.
  The corresponding generators, which include both conservative and dissipative evolutions, are
  given by the superposition of local generators in the Lindblad form.
  The associated infinite volume dynamics can be obtained as the strong limit of the finite volume dynamics.
  By regarding the interacting evolution as a perturbation of a non-interacting dissipative
  dynamics, we obtain a quantitative criterion that yields the uniqueness of the stationary state together with
  the exponential convergence of local observables.
  The analysis is based on suitable a priori bounds on the resolvent equation which yield
  quantitive estimates on the evolution of local observables.
\end{abstract}

\maketitle
\thispagestyle{empty}

\section{Introduction}
\label{s:1}

A quantum Markov semigroup $(\mathcal{P}_t)_{t\geq0}$ is a semigroup of completely positive operators on a $C^*$-algebra $\mc A$.
It has become the basic modelling tool to describe the non-unitary evolution of open quantum systems,
typical examples being interactions with thermal baths and measurement processes.
Its generator is commonly referred to as the Lindblad generator.
Both from a conceptual and an applied viewpoint, the ergodicity
properties of quantum semigroups are particularly relevant.
We refer to \cite{Al:Le,BP,Ev,Fa:Re,Fr},  for a general overview.

Under general conditions, finite dissipative quantum systems admit a unique stationary state $\pi$.
As in the case of classical Markov semigroups, a most relevant issue is to provide quantitative estimates on the speed of convergence to the stationary state.
More precisely, given an initial state $\mu$ one would like to deduce the
exponential convergence of its evolution $\mu\mathcal{P}_t$ to $\pi$;
a natural distance to quantify this convergence
is the trace distance, the non-commutative counterpart of the total
variation distance. By the quantum Pinsker inequality, see
e.g.\ \cite[Thm.11.9.5]{wilde}, the convergence in trace distance can be deduced from
the decay of the quantum relative entropy of
$\mu \mathcal{P}_t$ with respect to $\pi$.  For reversible semigroups, the latter issue has been recently pursued with different perspectives
\cite{Br:Ga:Ju,CM1,CM2,DR,Coreani_usciti}.  In particular, in \cite{CM1} the exponential
decay of entropy is deduced for both the Bose and Fermi
Ornstein-Uhlenbeck semigroups that describe the evolution
of non-interacting bosons and fermions.  As discussed in \cite{BDP}, the exponential convergence
in trace distance can also be deduced by spectral methods.

The present purpose is instead to investigate the ergodic properties of semigroups corresponding to infinite quantum lattice systems.
We refer to \cite{Bu,BHMKPSZ-2021,Cr:De:Sc,naaijkens,PEPS-2018} for some physically relevant examples.
According to the standard framework, a quantum lattice system is described by the so-called \emph{quasi-local} $C^*$-algebra $\mathcal{A}$ \cite{Br:Ro,Ru}.
We consider here both the cases in which $\mathcal{A}$ describes generic qudits and fermions.

A simple choice for the Lindblad generator is given by the superposition of local generators.
In other words we consider generators on $\mathcal{A}$ of the form
\begin{equation}\label{eq:prel}
  \mathcal{L}=\sum_{X\subset\subset\bb Z^d}L_X,
\end{equation}
where the sum is carried over the finite subsets of $\bb Z^d$ and the local Lindblad generator $L_X$ acts on $\mathcal{A}_X$, the subalgebra of the $X$-local operators.
While this locality assumption is taken for granted in the classical case, it can be questioned at the quantum level due to the intrinsic  non-locality of the theory.
However, a single site dissipation induces decoherence effects which support the locality assumption. 

In the simplest and most relevant case, the family $\{L_X\}_{X\subset\subset\bb Z^d}$ is translation covariant and has finite range.
For each $X\subset\subset\bb Z^d$ the $C^*$-algebra $\mathcal{A}_X$ is finite dimensional and $L_X$ can be prescribed according to the Lindblad-Gorini-Kossakowski-Sudarshan structure theorem \cite{Go:Ko:Su,Lindblad76}.
On the other hand, the right-hand side of \eqref{eq:prel} defines an unbounded operator on the $C^*$-algebra $\mathcal{A}$.
A preliminary issue is thus to show that $\mathcal{L}$ generates a semigroup on $\mathcal{A}$.
This is the dissipative counterpart to the existence of the Heisenberg flow for quantum lattice systems, see e.g.\  \cite[Thm.7.6.2]{Ru}, and it has been considered in \cite{NVZ2011}, where this semigroup is constructed as the strong limit of finite volumes semigroups.
We here pursue a somehow different approach.
We first define $\mathcal{L}$ on a suitable dense domain and then deduce sufficient conditions, holding in the translation covariant finite range case, to  apply the Lumer-Philips theorem.
We thus infer that the graph norm closure of $\mathcal{L}$ generates a strongly continuous semigroup $(\mathcal{P}_t)_{t\geq0}$.

We next turn to the discussion of the ergodic properties of  $(\mathcal{P}_t)_{t\geq0}$, the main point of the present analysis.
By a standard soft compactness argument, $(\mathcal{P}_t)_{t\geq0}$ has at least a stationary state.
On the other hand, for instance by considering classical stochastic Ising models, it is straightforward to exhibit examples for which the stationary state is not unique, see e.g. \cite{Li}.
If the interaction is small, we however expect uniqueness of the stationary state and exponential convergence of local observables.

We here deduce indeed perturbative criteria for the above conclusion.
We write $\mathcal{L}=\mathcal{L}_0+\mathcal{L}_1$ where $\mathcal{L}_0$ describes the evolution of non-interacting qudits, i.e.\ it has the form
\begin{equation*}
  \mathcal{L}_0
  =\sum_{x\in\bb Z^d}L_x^0,
\end{equation*}
for suitable translation covariant generators $L^0_x$ that act on $\mathcal{A}_{\{x\}}$.
The operator $\mathcal{L}_1$, that takes into account the interaction between the qudits, is then regarded as a perturbation.
As discussed in \cite{CM1,DR,Coreani_usciti}, a basic tool in the derivation of strong ergodic properties for the semigroup generated by $\mathcal{L}_0$ is the construction of operators $\{E_{x,h}\}$ on $\mathcal{A}$ satisfying the following intertwining relationship with $\mathcal{L}_0$
\begin{equation}\label{eq:itw}
  E_{x,h}\mathcal{L}_0-\mathcal{L}_0 E_{x,h}
  = -\lambda_h E_{x,h}
\end{equation}
for suitable $\lambda_h>0$.
When $\mathcal{A}$ is the fermionic $C^*$-algebra and $\mathcal{L}_0$ is the generator of the Fermi Ornstein-Uhlenbeck semigroup, the operators $\{E_{x,h}\}$ have been introduced in \cite{G}.
On the other hand, when $\mathcal{A}$ is a $C^*$-algebra describing generic qudits, we do not specify an explicit form for $L_x^0$ but we assume that it is self-adjoint with respect to the GNS inner product induced by its unique stationary state.
We then construct the operators $E_{x,h}$ from the spectral decomposition of $L_x^0$; accordingly $\{\lambda_h\}$ are the eigenvalues of $-L_x^0$.
To achieve the perturbative criterion for ergodicity of the semigroup generated by $\mathcal{L}=\mathcal{L}_0+\mathcal{L}_1$, we follow the argument for interacting classical lattice systems in \cite[Ch.I]{Li}.
The key step is the derivation of suitable a priori bounds on the resolvent equation
\begin{equation*}
  \lambda g= f+\mathcal{L}g,
  \qquad\qquad\lambda>0, \quad f,g\in\mathcal{A}.
\end{equation*}
By exploiting the intertwining relation \eqref{eq:itw}, we derive a quantitative bound on the locality of $g$ in terms of the locality of $f$.
Under a suitable smallness assumption  on the commutator $E_{x,h}\mathcal{L}_1-\mathcal{L}_1 E_{x,h}$, the a priori bound on the resolvent equation yields the uniqueness of the stationary state together with the exponential convergence of local observables.
We emphasize that the perturbative criteria here obtained rely neither on the explicit knowledge of stationary states nor on the self-adjointness of the generator $\mathcal{L}$ in \eqref{eq:prel}.

Whenever the perturbation criterion applies, we can deduce that for any state $\mu$ on $\mathcal{A}$ the sequence $\mu\mathcal{P}_t$ converges to the unique stationary state $\pi$ as $t\to+\infty$.
As in the classical case, it does not appear however possible to
obtain a quantitative bound uniform in $\mu$.
Indeed, this fails even for non-interacting systems.
On the other hand, if we restrict to translation covariant interactions and translation invariant state $\mu$, the above conclusion holds.
More precisely, we equip the set of translation invariant states on $\mathcal{A}$ with the specific quantum one-Wasserstein distance $w$ introduced in \cite{DD}.
We then show that $w(\mu\mathcal{P}_t,\pi)$ decays exponentially uniformly in $\mu$ whenever the perturbative criterion holds.
This statement appears to be novel even in the context of interacting classical lattice systems.

\bigskip
\noindent
{\bf Frequently used notation}

\begin{longtable}{ll}
  $\Pfin$ &finite subsets of $\bb Z^d$\\
  $H$ &single site Hilbert space\\
  $\mathcal{H}$ &quantum lattice system Hilbert space\\
  $A$ &single site $C^*$-algebra\\
  $\mathcal{B}(\mathcal{H})$ &bounded operators on $\mathcal{H}$\\
  $\tau_x$ &translation operator \eqref{eq:tx}\\
  $\|\cdot\|$ &operator norm in $\mathcal{B}(\mathcal{H})$\\
  $E_{x,h}$ &off-diagonal spectral projection associated to $L_0$ \eqref{eq:otimes}\\
  $\delta_x$ &local seminorm \eqref{deltaf}\\
  $\vertiii{\cdot}$ &seminorm on $\mathcal{A}^0$ \eqref{3norm}\\
  $\mathcal{A}^0$ &strictly local algebra\\
  $\mathcal{A}^1$ &closure of $\mathcal{A}^0$ with respect to $\|\cdot\|+\vertiii{\cdot}$\\
  $\mathcal{A}$ &quasi-local $C^*$-algebra, closure of $\mathcal{A}^0$ with respect to $\|\cdot\|$\\
  $\ind$ &identity in $\mathcal{A}$\\
  $\rho$ &reference state on $A$\\ 
  $\mathcal{S}$ &set of states on $\mathcal{A}$\\
  $\mathcal{S}_\tau$ &set of stationary states on $\mathcal{A}$\\
  $w$ &specific quantum one-Wasserstein distance \eqref{W=}\\
  $\pi(f)$ &duality between $\mathcal{S}$ and $\mathcal{A}$\\
  $L_0$ &single site Lindblad generator \eqref{eq:L0}\\
  $\lambda_1$ &spectral gap of $L_0$\\
  $\mathcal{L}_0$ &lattice unperturbed Lindblad generator \eqref{mcL0}\\
  $\mathcal{L}$ &lattice full Lindblad generator \eqref{mcL0}\\
  $\mathcal{I}$ &index set for the local jump operators\\
  $\mathcal{I}_0$ &index set for the unperturbed local jump operators\\
  $\imath$ &inclusion map $\mathcal{I}_0\hookrightarrow\mathcal{I}$\\
  $\chi$ &support map $\mathcal{I}\to\Pfin$\\
  $k_\alpha$ &local Hamiltonian\\
  $\ell_\alpha$ &local jump operator\\
  $\ell_\alpha^0$ &unperturbed local jump operator\\
  $L^1_\alpha$ &local perturbation of the Lindblad generator \eqref{L1a}\\
  $(\mathcal{P}_t)_{t\geq0}$ &semigroup generated by $\mathcal{L}$\\
  $\mu\mapsto\mu\mathcal{P}_t$ &dual action of $\mathcal{P}_t$ on $\mathcal{S}$ 
\end{longtable}

\section{
Interacting qudits: results
} 
\label{s:2}

Given a Banach space $\mathcal{X}$, we denote by $\|\cdot\|_\mathcal{X}$ the norm in $\mathcal{X}$ and by $\mc X'$ the dual of $\mc X$.
The identity operator on $\mc X$ is denoted by $\id_\mathcal{X}$ and the operator norm on the set of bounded linear operators on $\mathcal{X}$ by $\|\cdot\|_{\mathcal{X}\to\mathcal{X}}$.
According to the terminology in \cite{Al:Le,gustavo,Fa,Ch:Fa}, a quantum Markov semigroup is a $\sigma$-weak continuous semigroup on a von Neumann algebra.
For our purposes, it will be convenient to consider instead strongly continuous semigroups that as in \cite[Def. 4.39]{gustavo} will be referred to as Feller semigroups.
More precisely, given a unital $C^*$-algebra $\mathcal{A}$, a \emph{Quantum Feller Semigroup} (QFS) $(\mathcal{P})_{t\geq0}$ on $\mathcal{A}$ is a strongly continuous contraction semigroup on $\mathcal{A}$ (as a Banach space) such that for each $t\in[0,\infty)$ the linear operator $\mathcal{P}_t\colon\mathcal{A}\to\mathcal{A}$ is completely positive and satisfies $\mathcal{P}_t\ind=\ind$, where $\ind$ is the identity in $\mathcal{A}$.
The advantage of using strongly continuous semigroups rather than $\sigma$-weak continuous semigroups is the possibility to use Hille-Yoshida and Lumer-Philips theorems which are crucial for the present analysis.

\subsection{One-qudit unperturbed dynamics}

Let $H$ be a $(n+1)$-dimensional Hilbert space and $A$ a
$(N+1)$-dimensional unital $*$-subalgebra of the $C^*$-algebra $\mathcal{B}(H)$ of linear operators on $H$.  Since $H$ is finite dimensional, $A$ endowed
with the operator norm is closed and therefore it is a $C^*$-algebra.
Let $\rho$ be a faithful state on $A$, namely a continuous linear
functional on $A$ such that $\rho(\id_H)=1$ and $\rho(aa^*)>0$ for all $a\in A\backslash\{0\}$.  Denote by
$\langle\,\cdot\,,\cdot\,\rangle_\rho$ the GNS inner product on $A$
induced by $\rho$, i.e.\ $\langle a,b\rangle_\rho=\rho(a^* b)$,
$a,b\in A$.
The state $\rho$ will be regarded as a fixed reference state.
Let $(P^0_t)_{t\geq0}$, be a QFS on $A$ which is self-adjoint with
respect to $\langle\,\cdot\,,\cdot\,\rangle_\rho$ and denote by
$L_0\colon A\to A$ its generator.  We assume that $L_0$ has the form
\begin{equation}
  \label{eq:L0}
  L_0 
  =\sum_{j\in I_0}\big({\ell^0_j}^*[\,\cdot\,
  ,\ell^0_j]+[{\ell_j^0}^*,\,\cdot\, ]\ell_j^0\big)
\end{equation}
for some finite set $I_0$ and $\big\{\ell_j^0\big\}_{j\in I_0} \subset A$.
We refer to \cite[Thm.3.1]{CM1}, \cite[Thm.3]{Ali}, and
\cite[Lem. 2.2]{BDP} for the conditions on $\big\{\ell_j^0\big\}_{j\in I_0}$ corresponding to the self-adjointness of $L_0$ with respect to $\langle\,\cdot\,,\cdot\,\rangle_\rho$.

Since $\rho$ is a stationary state for $(P^0_t)_{t\geq0}$, the Kadison-Schwarz inequality implies
$\langle P^0_ta,P^0_ta\rangle_\rho\leq\langle a,a\rangle_\rho$.
Therefore $(P^0_t)_{t\geq0}$ is a self-adjoint strongly continuous
contraction semigroup on the Hilbert space
$(A,\langle\,\cdot\,,\cdot\,\rangle_\rho)$.  Hence $-L_0$ is positive
definite with respect to $\langle\,\cdot\,,\cdot\,\rangle_\rho$.  Let
$\id_H=e_0,e_1,\dots,e_N\in A$ be an orthonormal basis, with respect
to $\langle\,\cdot\,,\cdot\,\rangle_\rho$, of eigenvectors of $-L_0$,
with eigenvalues $0=\lambda_0\leq\lambda_1\leq\dots\leq\lambda_N$.  In
particular,
\begin{equation}\label{L0dec}
-L_0
=
\sum_{h=1}^N\lambda_h\,\langle e_h,\,\cdot\,\rangle_\rho\, e_h
\,.
\end{equation}
We observe that $\|e_0\|_{_{H\to H}}=1$ and set
$\eta:=\max_{h\in\{1,\dots,N\}}\|e_h\|_{_{H\to H}}$.

\subsection{Unperturbed dynamics of qudits}
\label{sec:udq}

We denote by $\bb Z^d$ the standard $d$-dime\-nsional lattice and by $\Pfin$ the countable family of its finite subsets, $\Pfin:=\big\{X\subset\bb Z^d\,\colon\,|X|<\infty\big\}$.
Referring to \cite{naaijkens} for a detailed exposition, we briefly recall the construction of the Hilbert space and the $C^*$-algebra describing infinitely many qudits.

Fix an orthonormal basis $\{\omega_i\}_{i=0}^{n}$ of $H$, where
$\omega_0$ represents the vacuum state.  Let $\Phi$ be the countable
set of functions $\phi:\,\bb Z^d\to\{0,1,\dots,n\}$ such that
$\phi(x)=0$ for all but finitely many $x\in\bb Z^d$.  We then let
$\mc H$ be the Hilbert space with orthonormal basis
$\{\Omega_{\phi}\}_{\phi\in \Phi}$.  For $\Lambda\subset\bb Z^d$ we
also consider the subspace $\mc H_\Lambda\subset\mc H$ spanned by
$\{\Omega_\phi\}_{\phi\in\Phi_\Lambda}$, where
$\Phi_\Lambda\subset\Phi$ is the subset of functions $\phi\in\Phi$
such that $\phi(x)=0$ if $x\not\in\Lambda$.  For $\Lambda\in\Pfin$, we
identify $\mc H_\Lambda\simeq H^{\otimes\Lambda}$ via
$\Omega_\phi\mapsto\otimes_{x\in\Lambda}\omega_{\phi(x)}$,
$\phi\in\Phi_\Lambda$.  Note that the construction of $\mathcal{H}$
depends on the choice of the vacuum $\omega_0$.

Denote by $\mc B(\mc H)$ the $C^*$-algebra of bounded operators on
$\mc H$ and by $\|\cdot\|$ the corresponding operator norm.  For
$\Lambda\subset\bb Z^d$ we have the canonical identification
$\Phi\simeq\Phi_\Lambda\times\Phi_{\Lambda^c}$, which induces
$\mathcal{H}\simeq\mathcal{H}_\Lambda\otimes\mathcal{H}_{\Lambda^c}$. As a consequence we have the
canonical embedding
$\mathcal{B}(\mathcal{H}_\Lambda)\subset\mathcal{B}(\mathcal{H})$.  For
$\Lambda\in\Pfin$, we set
$\mc A_\Lambda:=A^{\otimes\Lambda}\subset\mathcal{B}(H^{\otimes\Lambda})$,
which we identify with the $*$-subalgebra of
$\mathcal{B}(\mc H_\Lambda)\subset\mc B(\mc H)$, via the above identification
$H^{\otimes\Lambda}\simeq\mc H_\Lambda$.  Since $\mc A_\Lambda$ is
finite dimensional, it is a $C^*$-algebra when equipped with the norm
$\|\cdot\|$.  Set
$\mc A^0:=\bigcup_{\Lambda\in\Pfin}\mc A_\Lambda\subset\mc B(\mc H)$,
and let $\mc A$ be the norm closure of $\mc A^0$.  In particular,
$\mc A$ is a unital $C^*$-subalgebra of $\mc B(\mc H)$.  We emphasize
that, even in the case $A=\mathcal{B}(H)$, $\mc A$ is a proper subalgebra of
$\mc B(\mc H)$.  For example, for $x\in\bb Z^d$, consider the
translation operator $\tau_x\in\mathcal{B}(\mathcal{H})$,
$x\in\bb Z^d$, defined by
\begin{equation}
  \label{eq:tx}
  \tau_x(\Omega_\phi):=\Omega_{\tau_x\phi},
  \quad\text{where}\quad
  \tau_x\phi:=\phi(\,\cdot\,-x),
  \qquad
  \phi\in\Phi.
\end{equation}
As simple to check, $\tau_x$ does not belong to $\mc A$.
On the other hand, it induces the automorphism
$\Ad_{\tau_x}\colon\mathcal{A}\to\mathcal{A}$ by
$\Ad_{\tau_x}(f):=\tau_xf \tau_{-x}$.  In the terminology of \cite[\S6.2.4]{Ru}, the triple $(\mathcal{A},\bb Z^d,\Ad_\tau)$ is the
\emph{quasi-local} algebra describing the qudits on $\bb Z^d$.
The set of the states on $\mathcal{A}$ is denoted by $\mathcal{S}$.
A state $\pi\in \mc S$ is \emph{translation invariant} if
$\pi\big( \Ad_{\tau_x}(f) \big) =\pi(f)$ for any $x\in \bb Z^d$ and
$f\in\mc A$; the set of translation invariant states is denoted by
$\mc S_\tau$. 
A state $\pi\in\mathcal{S}$ is \emph{stationary} for the QFS
$(\mathcal{P}_t)_{t\geq0}$ if $\pi\mathcal{P}_t=\pi$ for any
$t\in[0,+\infty)$.
Hereafter, we denote by $\mu\mapsto\mu\mathcal{P}_t$ the dual action of $\mathcal{P}_t$ on $\mathcal{S}$.

The unperturbed dynamics is next defined by letting each qudit evolve
according to the generator $L_0$ in \eqref{L0dec}.
For
$x\in\bb Z^d$ we set $L^0_x=L_0\otimes\id_{\mc A_{\{x\}^\mathrm{c}}}$, that is regarded as an operator on $\mathcal{A}$.
The unperturbed generator is then informally given by
\begin{equation}\label{mcL0}
\mc L_0
= \sum_{x\in\bb Z^d} L^0_x
= \sum_{\alpha\in\mathcal{I}_0}
\big(
{\ell^0_{\alpha}}^*[\, \cdot \, ,\ell^0_{\alpha}] +
[{\ell_{\alpha}^0}^*,\,\cdot\, ]\ell_{\alpha}^0
\big)
\end{equation}
where, recalling \eqref{eq:L0}, $\mathcal{I}_0:=\bb Z^d\times I_0$ and
for $\alpha=(x,j)$ the operator $\ell_\alpha^0$ corresponds to
$\ell_j^0$ via the identification $\mathcal{A}_{\{x\}}\simeq A$.  For
future purposes, we denote by $\chi_0\colon\mathcal{I}_0\to\bb Z^d$
the projection $\chi_0(x,j)=x$.
We will show that the right-hand side of \eqref{mcL0} is well defined
on a suitable dense subset of $\mc A$ and its graph norm closure
generates a QFS on $\mc A$ which leaves the product state $\rho^{\otimes\bb Z^d}\in\mathcal{S}$ invariant.

Recalling the spectral decomposition \eqref{L0dec}, let $E_{x,h}$,
$x\in\bb Z^d$, $h\in\{0,1,\dots,N\}$, be the linear operators on $\mathcal{A}^0$ acting on monomials as
\begin{equation}
  \label{eq:otimes}
  E_{x,h}\big(\otimes_{y}f_y\big)
  = \langle e_h,f_x\rangle_\rho\id_H\otimes
  \big(\otimes_{y\neq x}f_y\big). 
\end{equation}
As we prove in Lemma \ref{lem:exh}, $E_{x,h}$ extends to a bounded
operator on $\mc A$.  Let also $e_{x,h}$ be the element of
$\mathcal{A}_{\{x\}}\simeq A$ corresponding to $e_h$.
Since $\{e_h\}$ is an orthonormal basis of $A$, for each
$x\in\bb Z^d$ and $f\in\mathcal{A}$
\begin{equation}
  \label{eq:spdec}
  f=
  \sum_{h=0}^N(E_{x,h}f)e_{x,h}.
\end{equation}

Introduce the seminorm $\vertiii{\cdot}$ on $\mc A^0$ by setting
\begin{equation}\label{3norm}
\vertiii{f} 
:=
\sum_{x\in\bb Z^d}
\sum_{h=1}^N
\|E_{x,h} f\|,
\end{equation}
where we emphasize that $E_{x,0}$ does not appear on the right-hand
side.
We interpret $\sum_{h=1}^N\|E_{x,h} f\|$ as a measure of the dependence of $f$ on the qudit at site $x\in \bb Z^d$.
In particular, $\vertiii{f}=0$ if and only if $f$ is a scalar
multiple of the identity.  Let $\mc A^1$ be the closure of
$\mc A^0$ with respect to the norm $\|\cdot\|+\vertiii{\cdot}$.
Clearly, $\mc A^0\subset\mc A^1\subset\mc A$.
Let $\ell_1(\bb Z^d)$ be the Banach space of summable real sequences indexed by $\bb Z^d$.
For $f\in\mc A^1$, set
\begin{equation}
  \label{deltaf}
  \delta(f) =
  \big(\delta_x(f)\big)_{x\in\bb Z^d}
  :=\Big(\sum_{h=1}^N\|E_{x,h}f \|\Big)_{x\in\bb Z^d}\in\ell_1(\bb Z^d) \,,
\end{equation}
so that $\vertiii{f}=\sum_{x}\delta_x(f)=\|\delta(f)\|_{\ell_1(\bb Z^d)}$.

\subsection{Dynamics of interacting qudits}\label{sec:diq}

The dynamics of the qudits is defined by an unbounded Lindblad
generator on $\mc A$ given by the sum of local generators.  More
precisely, we fix a countable set $\mathcal{I}$ and a map
$\chi\colon \mathcal{I}\to\Pfin$ such that $|\chi^{-1}(X)|<+\infty$
for any $X\in\Pfin$.  We then consider the informal generator on
$\mc A$ given by
\begin{equation}
  \label{a.1}
  \mc L
  = \sum_{\alpha\in\mathcal{I}} \big(i [k_{\alpha},\,\cdot\,]
  +\ell_{\alpha}^*[\,\cdot\,,\ell_{\alpha}]
  +[\ell_{\alpha}^*, \,\cdot\,]\ell_{\alpha}\big)
\end{equation}
for some self-adjoint $k_\alpha \in\mathcal{A}_{\chi(\alpha)}$ and
$\ell_{\alpha}\in\mathcal{A}_{\chi(\alpha)}$.  Note that, by setting
$K_X=\sum_{\alpha\in\chi^{-1}(X)}k_{\alpha}$ and
\begin{equation}\label{eq:LX}
  L_X
  =i[K_X, \,\cdot\,] + \sum_{\alpha\in\chi^{-1}(X)}
  \big(\ell_{\alpha}^*[\,\cdot\,,\ell_{\alpha}]+[\ell_{\alpha}^*,
  \,\cdot\,]\ell_{\alpha}\big)
\end{equation}
then $\mathcal{L}$ has the form \eqref{eq:prel}.
If $A=\mathcal{B}(H)$ so that $\mathcal{A}_X\simeq\mathcal{B}(\mathcal{H}_X)$, by the Lindblad-Gorini-Kossakowski-Sudarshan structure theorem \cite{Go:Ko:Su,Lindblad76}, the right-hand side of \eqref{eq:LX} is the general form of a Lindblad generator on $\mathcal{A}_X$.

The family $\big\{k_\alpha,\ell_\alpha\big\}_{\alpha\in\mathcal{I}}$ has \emph{finite range} if there exists $R\in[0,\infty)$ such that $k _\alpha=\ell_\alpha=0$ whenever $\diam(\chi(\alpha))>R$.
The family $\big\{k_\alpha,\ell_\alpha\big\}_{\alpha\in\mathcal{I}}$ is \emph{translation covariant} if there exists an action of the abelian
group $\bb Z^d$ on $\mathcal{I}$, denoted by
$(x,\alpha)\mapsto x+\alpha$, satisfying
$\chi(x+\alpha)=x+\chi(\alpha)$, such that
$\Ad_{\tau_x}(k_\alpha) = k_{x+\alpha}$ and
$\Ad_{\tau_x}(\ell_\alpha) = \ell_{x+\alpha}$.
As we next state, under suitable conditions
the operator $\mc L$ in \eqref{a.1} is well defined
on the dense subset $\mc A^1\subset \mc A$.

\begin{lemma}\label{lem:LX}
  If
  \begin{equation}
    \label{eq:C0}
    C_0
    :=2\eta \sup_{x\in\bb Z^d} \sum_{\alpha\colon\!\chi(\alpha) \ni
      x}(\|k_\alpha\|+2\|\ell_{\alpha}\|^2) < \infty
  \end{equation}
  then for each $f\in\mathcal{A}^1$ the series defining $\mathcal{L}f$
  converges in $\mathcal{A}$ and
  $\|\mathcal{L}f\|\leq C_0\vertiii{f}$.
\end{lemma}

\subsection{Main results}\label{sec:mai}

The first result on interacting qudits establishes the existence of
the dynamics associated to the generator $\mc L$ introduced in
\eqref{a.1}. It will be convenient to write it as a perturbation of
the dynamics of non-interacting qudits as defined by \eqref{mcL0}. Namely,
we let $\mc L =\mc L_0 +\mc L_1$ where we recall that both $\mc L$ and
$\mc L_0$ have been constructed from local Lindblad generators.

This decomposition is achieved by decomposing $\ell_\alpha=\ell^0_{\alpha_0}+\ell^1_{\alpha_0}$ for $\alpha_0\in\mathcal{I}_0$ and some $\alpha\in\mathcal{I}$.
More precisely, fix an injective map
$\imath\colon\mathcal{I}_0\to\mathcal{I}$ such that
$\chi_0(\alpha_0)\in \chi(\imath(\alpha_0))$, $\alpha_0\in\mathcal{I}_0$.
Setting $\ell^1_{\alpha_0}=\ell_{\imath(\alpha_0)}-\ell^0_{\alpha_0}$,
$\alpha_0\in\mathcal{I}_0$, we have
$\mathcal{L}_1=\sum_{\alpha\in\mathcal{I}}L^1_\alpha$, where
\begin{equation}
  \label{L1a}
  L_\alpha^1= i [k_{\alpha},\,\cdot\,]+
\begin{cases}
    \begin{split}
      & {\ell^0_{\alpha_0}}^*[\,\cdot\,,\ell^1_{\alpha_0}]
      + [{\ell^0_{\alpha_0}}^*, \,\cdot\,]\ell^1_{\alpha_0}\\
      &\ +{\ell^1_{\alpha_0}}^*[\,\cdot\,,\ell^0_{\alpha_0}] +
      [{\ell^1_{\alpha_0}}^*, \,\cdot\,]\ell^0_{\alpha_0}\\ 
      &\ +{\ell^1_{\alpha_0}}^*[\,\cdot\,,\ell^1_{\alpha_0}]
      +[{\ell^1_{\alpha_0}}^*, \,\cdot\,]\ell^1_{\alpha_0}
    \end{split}
    &\text{if $\alpha=\imath(\alpha_0)\in\imath(\mathcal{I}_0)$,}\\
    \\
    \ell_{\alpha}^*[\,\cdot\,,\ell_{\alpha}] + [\ell_{\alpha}^*,
  \,\cdot\,]\ell_{\alpha} & \text{if $\alpha\not\in\imath(\mathcal{I}_0)$.}
  \end{cases}
\end{equation}
The strength of the perturbation $\mc L_1$ is measured by  
\begin{equation}
  \label{a.6}
  M := \sup_{y\in \bb Z^d} \sum_{x\in \bb Z^d} \theta_{x,y}
\end{equation}
where $\theta_{x,y}:=2N\eta(\theta_{x,y}^0+\theta_{x,y}^1)$, in which
\begin{equation*}
   \begin{split}
     &\theta_{x,y}^0:=
     \sum_{\substack{\alpha_0\in \mc I_0 \\ \chi_0(\alpha_0)= y}}
    \Big( \big( 1 +\eta^{2} \delta_{x,y}\big) \delta_x( k_{\imath(\alpha_0)})\\
    &\qquad\qquad + 2\big( \eta^{2}+ \delta_{x,y}\big)
    \Big[
    \delta_x\big( {\ell_{\alpha_0}^0} ^{\!\!\! *} \big)
      \, \big\|\ell_{\alpha_0}^1 \big\|
      +\big\| {\ell_{\alpha_0}^0}^{\!\!\! *} \big\|
          \, \delta_x\big( {\ell_{\alpha_0}^1} \big) + \delta_x\big( {\ell_{\alpha_0}^1} ^{\!\!\! *} \big)
      \, \big\|\ell_{\alpha_0}^0 \big\|
    \\
    &
     \qquad\qquad 
      +\big\| {\ell_{\alpha_0}^1}^{\!\!\! *} \big\|
      \, \delta_x\big( {\ell_{\alpha_0}^0} \big)
      + \delta_x\big( {\ell_{\alpha_0}^1} ^{\!\!\! *} \big)
      \, \big\|\ell_{\alpha_0}^1 \big\|
      +\big\| {\ell_{\alpha_0}^1}^{\!\!\! *} \big\|
      \,\delta_x\big( {\ell_{\alpha_0}^1} \big)
     \Big]\Big)
  \end{split}
\end{equation*}
and
\begin{equation*}
  \begin{split}
     \theta_{x,y}^1:=
     \sum_{\substack{\alpha\in \mc I\setminus \imath(\mc I_0)
         \\ \chi(\alpha)\ni y}}
    \Big(\big( 1 +\eta^{2} \delta_{x,y}\big) \delta_x( k_{\alpha})
    +2\big( \eta^{2}+ \delta_{x,y}\big)
    \big[\delta_x(\ell_{\alpha}^*) \,  \|\ell_{\alpha}\|
     +\|\ell_{\alpha}^*\| \, \delta_x(\ell_{\alpha}^*) \big]
     \Big).
  \end{split}
\end{equation*}
We observe that both $C_0$ in \eqref{eq:C0} and $M$ in \eqref{a.6} are finite whenever the family $\{k_\alpha,\ell_\alpha\}_{\alpha\in\mc I}$ is translation covariant and has finite range.
For $\Lambda\in\Pfin$ we denote by $\mathcal{L}_\Lambda$ the bounded Lindblad generator on $\mathcal{A}$ defined by
\begin{equation}\label{eq:Lfin}
  \mathcal{L}_\Lambda
  :=\sum_{X\subset\Lambda}L_X
\end{equation}
and by $(\mathcal{P}_t^\Lambda)_{t\geq0}$ the corresponding QFS.
We finally recall that $\lambda_1$ is the spectral
gap of the unperturbed one-qudit generator $L_0$.

\begin{theorem}\label{thm1}
  Assume  $C_0,M<\infty$ and consider $\mc L$ as an
  operator on $\mc A$ with domain $\mc A^1$. Then
  \begin{enumerate}[(i)]
  \item the graph norm closure $\bar{\mc L}$ of $\mc L$ generates a
    QFS $(\mc P_t)_{t\geq0}$ on $\mc A$;
  \item $\mc A^1$ is a core for $\bar{\mc L}$;
  \item for each $t\geq0$ the operator $\mathcal{P}_t$ is the strong limit of $\mathcal{P}_t^\Lambda$ as $\Lambda\uparrow\bb  Z^d$;
  \item $\vertiii{\mc P_tf}\leq e^{(M-\lambda_1)t}\vertiii{f}$,
    for any $f\in\mc A^1$ and $t\geq0$;
  \item the QFS $(\mathcal{P}_t)_{t\geq0}$ has at least one stationary state.
    A state $\pi\in\mathcal{S}$ is stationary if and only if $\pi(\mathcal{L}f)=0$ for any $f\in\mathcal{A}^0$.
\end{enumerate}
\end{theorem}

As already mentioned, existence of the infinite volume dynamics has been elegantly proven in \cite{NVZ2011} by deriving the Lieb-Robinson bounds
\cite{Li:Ro,Na:Si} and constructing $(\mathcal{P}_t)_{t\geq0}$ as the strong limit of finite volume semigroups.
The assumptions in \cite{NVZ2011} are expressed in terms of the completely bounded norm of the local Lindblad generator $L_X$ in \eqref{eq:prel}.
We note that, since $L_X$ acts on a finite dimensional space it is completely bounded \cite[Prop.~8.11]{Pa}, however its completely bounded norm may grow with $|X|$, as \cite[Ex.~1.8]{Pa} shows.
In contrast, the assumption $M<\infty$ in Theorem~\ref{thm1} quantifies, using $\delta_x$ as introduced in \eqref{deltaf}, the dependence of the jump operators on the qudits.
Our assumption is therefore a condition of ``summable gradients'' rather than the condition of ``summable norms'' in \cite{NVZ2011} and it does not introduce additional dependence on $X$ due to the use of the completely bounded norm.
Furthermore, Theorem~\ref{thm1}.(i)--(ii) characterizes the generator of $(\mathcal{P}_t)_{t\geq0}$ and this allows the infinitesimal characterization of stationary states in Theorem~\ref{thm1}.(v).

In Theorem~\ref{thm1} we can take
$\mathcal{L}_0=0$, letting for example $\rho$ be the normalized trace
on $H$ and $\{e_h\}_{h=0}^N$ be any orthonormal basis with respect to
the normalized Hilbert-Schmidt inner product.
Hence, in the particular case in which $\mathcal{L}=i[\mc K,\,\cdot\,]$ is the Heisenberg operator associated
with the (informal) infinite volume Hamiltonian
$\mc K=\sum_{X\in \Pfin }K_X$, with $K_X$ as defined below
\eqref{a.1}, Theorem~\ref{thm1} yields the existence of a one-parameter group of
automorphisms $(\mathcal{U}_t)_{t\in\bb R}$ describing the Heisenberg
evolution on $\mathcal{A}$.
In this case, the assumptions of Theorem~\ref{thm1} are analogous to classical
conditions in the literature, see e.g.\ \cite[Thm. 7.6.2]{Ru}.

The next result, which provides a perturbative criterion for the
ergodicity of the QFS $(\mathcal{P}_t)_{t\geq0}$, depends instead on
the non-vanishing of $\mathcal{L}_0$ and more precisely on the
existence of a strictly positive spectral gap for the one-qudit
dynamics.
If the family $\{k_\alpha,\ell_\alpha\}_{\alpha\in\mathcal{I}}$ has finite range the corresponding unique stationary state has exponentially decaying correlations.

\begin{theorem}\label{thm2}
Assume \eqref{eq:C0} and $M<\lambda_1$.
Then
\begin{enumerate}[(i)]
\item
the QFS $(\mc P_t)_{t\geq0}$ has a unique stationary state $\pi$;
\item
for any $f\in\mc A^1$ and $t\geq0$
\begin{equation*}
  \big\|\mc P_tf-\pi(f)\ind\big\|
\leq
\frac {C_0}{\lambda_1-M} \, e^{-(\lambda_1-M)t} \, \vertiii{f};
\end{equation*}

\item
  if furthermore $\{k_\alpha,\ell_\alpha\}_{\alpha\in\mc I}$ has
  finite range then  
  there exist $C ,\zeta>0$ such that for any $\Lambda_1,\Lambda_2\in\Pfin$ and any $f_1\in\mc
  A_{\Lambda_1}$,
  $f_2\in\mc A_{\Lambda_2}$ 
  \begin{equation*}
    \big|\pi(f_1f_2)-\pi(f_1)\pi(f_2)\big|
    \leq
    C
    e^{-\zeta\,\mathrm{dist}(\Lambda_1,\Lambda_2)}
    \big(\|f_1\|+\vertiii{f_1}\big)\big(\|f_2\|+\vertiii{f_2}\big)  .
  \end{equation*}
\end{enumerate}
\end{theorem}

The condition $M<\lambda_1$ can be explicitly checked for specific
models, we refer to the Sections \ref{sec:qss} and \ref{sec:app} for the cases of quantum spin systems and of the $XYZ$-model
with site dissipation.
In this respect, items (ii) and (iii) provide quantitative bounds on
the speed of convergence to the stationary state and on the spatial
decay of correlations for the stationary state.

By the density of $\mathcal{A}^1$ in $\mathcal{A}$, Theorem~\ref{thm2}(ii) implies that
for each $\mu\in\mathcal{S}$ the sequence $\mu\mathcal{P}_t$ converges
weakly* to $\pi$.  Even in the case in which
$\mathcal{L}=\mathcal{L}_0$, it does not appear however possible to
obtain a quantitative bound on this convergence uniformly in
$\mu\in\mathcal{S}$.  On the other hand, as we next discuss, if we
restrict to translation covariant interactions and translation
invariant $\mu\in\mathcal{S}$, there is a natural distance on the set
of translation invariant states such that the distance between
$\mu\mathcal{P}_t$ and $\pi$ vanishes exponentially uniformly in
$\mu$.  More precisely, the final topic that we discuss is the
exponential convergence of the QFS $(\mathcal{P}_t)_{t\geq0}$ to the
unique stationary state $\pi$ in term of the specific quantum
one-Wasserstein distance introduced in \cite{DD}, which is the
non-commutative counterpart of the Ornstein $\bar d$ distance on the
set of translation invariant probabilities.
Given $\Lambda\in\Pfin$ let $\|\,\cdot\,\|_{W_\Lambda}$ be the norm on the space $\mathcal{O}_\Lambda^0$ of self-adjoint and
traceless elements in $\mathcal{A}_\Lambda$ defined by
\begin{equation*}
  \|\Delta\|_{W_\Lambda} 
  := \frac{1}{2}\inf\Big\{\sum_{x\in\Lambda}
  \|\Delta^{(x)}\|_{\Lambda, \tr}\colon \Delta^{(x)} \in
  \mathcal{O}_\Lambda^0,\, \tr_{\{x\}}\Delta^{(x)}=0,\,
  \sum_{x\in\Lambda}\Delta^{(x)}=\Delta \Big\}, 
\end{equation*}
where $\|\cdot\|_{\Lambda,\tr}$ is the trace norm on
$\mathcal{A}_\Lambda$, i.e.\ $\|f\|_{\Lambda,\tr}=\tr(\sqrt{ff^*})$, and
$\tr_{\{x\}}\colon\mathcal{A}_\Lambda\to\mathcal{A}_{\Lambda\setminus\{x\}}$
denotes the partial trace on $\mathcal{H}_{\{x\}}$.
Denoting by $\mathcal{S}_\Lambda$ the set of states on
$\mathcal{A}_\Lambda$, the \emph{quantum one-Wasserstein
  distance} $W_\Lambda$ on $\mathcal{S}_\Lambda$ is defined by $W_{\Lambda}(\mu,\nu):=\|\mu-\nu\|_{W_\Lambda}$.
Here we have identified $\mathcal{S}_\Lambda$ with the positive elements in $\mathcal{A}_\Lambda$ with unit trace.
Recalling that
$\mathcal{S}_\tau$ is the set of translation invariant states on
$\mathcal{A}$, as proven in \cite[Prop. 4.1]{DD}, the \emph{specific
  quantum one-Wasserstein distance} is the distance on
$\mathcal{S}_\tau$ defined by
\begin{equation}\label{w=}
  w(\mu,\nu)
  : =\sup_{\Lambda\in\Pfin}\frac{1}{|\Lambda|}W_\Lambda(\mu_\Lambda,\nu_\Lambda)
  =\lim_{\Lambda\uparrow\bb
    Z^d}\frac{1}{|\Lambda|}W_\Lambda(\mu_\Lambda,\nu_\Lambda), 
\end{equation}
where $\mu_\Lambda$ denotes the restriction of the state
$\mu\in\mathcal{S}_\tau$ to a state on $\mathcal{A}_\Lambda$.  Observe
that the topology on $\mathcal{S}_\tau$ induced by $w$ is finer than
the weak* topology.

\begin{theorem}\label{thm3}
  Assume \eqref{eq:C0}, $M<\lambda_1$, and that $\{k_\alpha,\ell_\alpha\}_{\alpha\in\mc I}$ is
  translation covariant.
  Then the unique stationary state $\pi$ of the QFS $(\mc P_t)_{t\geq0}$ is
  translation invariant and there exists a constant $C>0$ such that
  for any $\mu\in\mathcal{S}_\tau$ and $t\geq0$
  \begin{equation*}
    w(\mu \mc P_t,\pi)
    \leq C e^{-(\lambda_1-M)t}.
  \end{equation*}
\end{theorem}

The proof of both Theorems \ref{thm1} and \ref{thm2} follows the
strategy used in the construction of the Markov semigroup describing the
evolution of interacting classical lattice systems \cite[Ch.~I]{Li}.
The key ingredient is an a priori bound on the
resolvent equation $(\lambda-\mathcal{L})g=f$ showing that
$g\in\mathcal{A}^1$ whenever $f\in\mathcal{A}^1$.
By approximating $\mathcal{L}$ with the
finite volume generator $\mathcal{L}_\Lambda$ and using the Lumer-Phillips theorem, this a
priori bound implies the existence of the infinite volume dynamics.
As in the commutative case, when $M<\lambda_1$, the a priori bound
obtained on the resolvent equation actually implies that the seminorm
$\vertiii{\,\cdot\,}$ is exponentially contracted by the QFS
$(\mathcal{P}_t)_{t\geq0}$.  By routine arguments, this yields the
exponential convergence to equilibrium stated in
Theorem~\ref{thm2}(ii).  The exponential decay of spatial correlation 
at equilibrium in Theorem~\ref{thm2}(iii) follows from Theorem~\ref{thm2}(ii) and the Lieb-Robinson bounds \cite{Li:Ro,NVZ2011}, also referred to as ``finite speed of propagation'' in the commutative setting \cite[\S I.4]{Li}.
While Theorem~\ref{thm3} is a straightforward consequence of
Theorem~\ref{thm2}, its formulation appears novel also in the
context of interacting classical lattice systems.

From a technical viewpoint, in the commutative case discussed in
\cite[Ch.~I]{Li} the seminorm $\vertiii{f}$ is defined in
terms of the oscillations of $f$ at the sites $x\in \bb Z^d$ while
here it is adapted to the unperturbed dynamics. Correspondingly, while in  
\cite[Ch.~I]{Li} the strength of the unperturbed dynamics is specified
by a Doeblin condition on the transition rates, here it is measured by the spectral gap $\lambda_1$ of the unperturbed one-qudit generator.
Accordingly, a crucial input for the derivation of the a priori bound on the resolvent equation is the intertwining relationship \eqref{eq:itw} for the unperturbed generator.

\section{
Interacting qudits: proofs
} 
\label{s:2-c}

In this section we prove Theorems \ref{thm1}, \ref{thm2}, and \ref{thm3}.

\subsection{Semigroup generation}
\label{sub:1}

Recalling the definition of $\eta$ below \eqref{L0dec} we first show that $\|E_{x,h}\|_{\mathcal{A}\to\mathcal{A}}\leq\eta$.

\begin{lemma}\label{lem:exh}
  For each $x\in\bb Z^d$ and $h\in\{0,1,\dots,N\}$, the operator
  $E_{x,h}$ defined by \eqref{eq:otimes} extends to a bounded operator
  on $\mc A$.  In fact, for each $f\in\mc A$ we have
  $\|E_{x,0}f\|\leq \|f\|$ and $\|E_{x,h}f\|\leq \eta\,\|f\|$,
  $h\in\{1,\ldots, N\}$. 
\end{lemma}

\begin{proof}
  By the density of $\mc A^0$ in $\mc A$, it suffices to prove the
  stated inequalities for $f\in\mc A_\Lambda$ with $\Lambda\in\Pfin$.
  We first observe that, by the very definition of $E_{x,h}$, we have
  $E_{x,h}f=E_{x,0}(e^*_{x,h}f)$.
  On the other hand, for $g\in\mathcal{A}_\Lambda$ we have
  $E_{x,0}g=\tr_{\{x\}}(\rho g)\otimes\id_{\mathcal{H}_{ \{x\}}}$ in
  which
  $\tr_{\{x\}}\colon
  \mathcal{A}_\Lambda\to\mathcal{A}_{\Lambda\setminus \{x\}}$ is the
  partial trace on $H\simeq\mathcal{H}_{ \{x\}}$.  Since
  $g\mapsto \tr_{\{x\}}(\rho g)$ is completely positive and unital,
  the Kadison-Schwarz inequality implies $\|E_{x,0}g\|\leq\|g\|$.
  As $e_{x,0}=\ind$ and $\|e_{x,h}\|\le \eta$, $h\in\{1,\ldots,N\}$,
  this bound yields the claim.
\end{proof}

\begin{proof}[Proof of Lemma \ref{lem:LX}]
  The statement is a direct consequence of the following bound.
  For each $X\in\Pfin$, $u\in\mathcal{A}_X$, and $f\in\mathcal{A}^1$,
  \begin{equation}\label{eq:kbound}
    \|[u,f]\|
    \leq2\eta \|u\| \sum_{x\in X}\delta_x(f).
  \end{equation}
  To prove this inequality, given $x\in\bb Z^d$ define the operator $F_x\colon\mathcal{A}\to\mathcal{A}$ by
  \begin{equation}\label{eq:Fx=}
       F_xf
    =\sum_{h=1}^N (E_{x,h} f) e_{x,h},
  \end{equation}
  so that \eqref{eq:spdec} can be recast as $f=E_{x,0}f+F_xf$.
  Enumerating the elements of $X=\{x_1,\dots,x_m\}$ and using recursively this identity,
  \begin{equation}\label{eq:Fx}
    f
    =\Big(\prod_{j=1}^m E_{x_j,0}\Big)f +
    \sum_{j=1}^m\Big(\prod_{i< j} E_{x_i,0}\Big) \,
    F_{x_j} f.
  \end{equation}
  Since the first term on the right-hand side commutes  with $u\in\mathcal{A}_X$, the bound \eqref{eq:kbound} follows by observing that
  $\|[f_1,f_2]\|\leq 2\|f_1\|\|f_2\|$ and
  $\|E_{x,0}\|_{\mathcal{A}\to\mathcal{A}}=1$, $\|F_xf\|\leq\eta\delta_x(f)$, $x\in\bb Z^d$. 
\end{proof}

In order to construct the QFS generated by the operator $\mathcal{L}$
defined in \eqref{a.1}, we shall use the terminology and the results of
\cite[\S X.8]{RS2}.  In particular, a densely defined operator $T$ on
a Banach space $\mc X$ with domain $\mc D$ is \emph{accretive} if for
each $x\in\mc D$ there exists $\wp\in\mc X^\prime$ such that $\|\wp\|_{\mc X^\prime}=1$, $\wp(x)=\|x\|_{\mc X}$, and
$\re(\wp(Tx))\geq0$.
By the finite dimensional theory of quantum Markov semigroups, the bounded operator $L_X$, $X\in\Pfin$, as
defined in \eqref{eq:LX}, generates a QFS on the finite dimensional $C^*$-algebra $\mc A_X$.
The Lumer-Philips theorem \cite[Thm.~X.48]{RS2} thus implies that $-L_X$ is accretive and therefore also $-\mathcal{L}$, with domain $\mathcal{A}^1$, is accretive.
For the sake of completeness, we however next provide a direct proof of the
accretivity of $-\mc L$.

\begin{lemma}\label{lem:accr}
  The operator $-\mc L$ with domain $\mc A^1$ is accretive.
\end{lemma}

\begin{proof}
  By the definition of $\mc A^1$, Lemma \ref{lem:LX}, and the Banach-Alaoglu
  theorem, it is enough to show that for each $f\in\mc A_{\Lambda}$,
  with $\Lambda\in\Pfin$, there exists $\wp\in\mc A^\prime$ such that
  \begin{equation}\label{accretive}
    \|\wp\|_{\mc A^\prime}=1,\qquad
    \wp(f)=\|f\|,\qquad
    \re(\wp(\mc Lf))\leq0.
  \end{equation}
  Since $\mc A_{\Lambda}$ is finite dimensional, there exists
  $\xi\in\mc H$ with $\|\xi\|_{\mc H}=1$, which is eigenvector of $ff^*$ with
  maximal eigenvalue: $(ff^*)\xi=\|f\|^2 \xi$.  Let
  $\wp\in\mc A^\prime$ be defined by
  \begin{equation*}
    \wp(g) = \frac{(\xi,gf^*\xi)_{\mathcal{H}}}{\|f\|} 
  \end{equation*}
  where $(\cdot,\cdot)_{\mathcal{H}}$ denotes the inner product in $\mc H$.
  We claim that this functional fulfil the three conditions in
  \eqref{accretive}.
  The second condition holds trivially.  The first one follows from
  the second and the bound 
  \begin{equation*}
    |\wp(g)| \leq \frac{\|gf^*\xi\|_{\mc H}} {\|f\|} \leq
    \frac{\|gf^*\|}{\|f\|} \leq \frac{\|g\|\|f^*\|}{\|f\|}
    =\|g\| .
  \end{equation*}
  Recalling \eqref{a.1}, in order to prove the third condition in
  \eqref{accretive} it suffices to show that for each self-adjoint
  $k\in \mc A$ and $u\in\mc A$
  \begin{equation}
    \label{z.1}
    \re \wp\big( i [k,f] + [u^*,f]u+u^*[f,u] \big) \le 0.
  \end{equation}
  For the first term we have
  \begin{equation*}
    \re\wp(i[k,f])
    =
    -\frac1{\|f\|}
    \im(\xi,[k,f]f^*\xi)_\mathcal{H}
    =
    \frac1{\|f\|}
    \im\big(
    (f^*\xi,kf^*\xi)_\mathcal{H}
    -
    \|f\|^2(\xi,k\xi)_\mathcal{H}
    \big) =0 
  \end{equation*}
  since $k$ is self-adjoint. On the other hand,
  \begin{align*}
    & \re\wp([u^*,f]u+u^*[f,u])
      =
      \re\wp(2u^*fu-u^*uf-fu^*u) \\
    & \qquad
      = 
      \frac1{\|f\|}
      \re\big(\xi,(2u^*fu-u^*uf-fu^*u)f^*\xi\big)_\mathcal{H} \\
    & \qquad =
      \frac1{\|f\|}
      \re\Big(
      2(\xi,u^*fuf^*\xi)_\mathcal{H}
      -
      (\xi,u^*uff^*\xi)_\mathcal{H}
      -
      (\xi,fu^*uf^*\xi)_\mathcal{H}
      \Big) \\
    & \qquad =
      \frac1{\|f\|}
      \re\Big(
      2(f^*u\xi,uf^*\xi)_\mathcal{H}
      -
      \|f\|^2\|u\xi\|^2_{\mc H}
      -
      \|uf^*\xi\|^2_{\mc H}
      \Big) \\
    & \qquad \leq
      \frac1{\|f\|}
      \big(
      \|f^*u\xi\|^2_{\mc H}
      -
      \|f\|^2 \|u\xi\|^2_{\mc H}
      \big)
      \leq0,
  \end{align*}
  where we used Cauchy-Schwarz in the second last step.  The proof of
  \eqref{z.1} is thus completed.
\end{proof}

We next prove the intertwining relationship between the
unperturbed generator $\mc L_0$ and the operators $E_{x,h}$ defined by
\eqref{eq:otimes}. 

\begin{lemma}\label{lem24}
  For each $f\in{\mc A}^1$, $x\in{\bb Z}^d$, and
  $h\in\{0,1,\dots,N\}$,
  \begin{equation*}
    E_{x,h} {\mc L}_0 f-{\mc L}_0E_{x,h} f = -\lambda_hE_{x,h}f .  
  \end{equation*}
\end{lemma}

\begin{proof}
  By linearity and density it is enough to prove the statement for a monomial, $f=\otimes_y f_y$.
  Recalling \eqref{mcL0}, from the spectral decomposition \eqref{L0dec} and the definition \eqref{eq:otimes} of $E_{x,h}$ we deduce
  \begin{align*}
    E_{x,h}\mathcal{L}_0f
    &=-\lambda_h\langle e_h,f_x\rangle_\rho\big(\otimes_{y\neq x}f_y\big)\\
   &\quad-\sum_{z\neq x}\sum_{k=1}^N\lambda_k\langle e_k,f_z\rangle_\rho\langle e_h,f_x\rangle_\rho\big(\otimes_{y\neq x,z}f_y\big)\otimes e_{z,k},\\
        \mathcal{L}_0E_{x,h}f
    &=-\sum_{z\neq x}\sum_{k=1}^N\lambda_k\langle e_k,f_z\rangle_\rho\langle e_h,f_x\rangle_\rho\big(\otimes_{y\neq x,z}f_y\big)\otimes e_{z,k}.
  \end{align*}
  The statement follows.
\end{proof}

The following lemma provides the key estimate in realizing the
generator $\mc L$ as a perturbation of $\mc L_0$. Recall that
$\mc L_1=\sum_{\alpha\in\mc I} L^1_\alpha$ with
$L^1_\alpha$ defined in \eqref{L1a} and that $\delta(f)\in \ell_1(\bb Z^d)$ has
been defined in \eqref{deltaf}.

\begin{lemma}\label{lem246}
  For each $\alpha\in\mc I$, $f\in{\mc A}^1$, $x\in{\bb Z}^d$, and
  $h\in\{1,\dots,N\}$,
  \begin{equation*}
    \big\|E_{x,h} L^1_\alpha f- L^1_\alpha E_{x,h} f\big\| \leq
    \sum_{y\in \bb Z^d} \theta_{x,y}(\alpha) \delta_y(f)
  \end{equation*}
  where
  \begin{equation*}
    \begin{split}
      &\theta_{x,y}(\alpha):=
      2\eta (1+ \eta^2\delta_{x,y})\delta_x( k_\alpha)
      \\
      &
      \;\;
      + 4\eta^3(1 +\eta^{-2} \delta_{x,y})
      \!\times\!
   \begin{cases}
           \begin{split}
        &
      \delta_x\big( {\ell_{\alpha_0}^0} ^{\!\!\! *} \big)
      \, \big\|\ell_{\alpha_0}^1 \big\|
      +\big\| {\ell_{\alpha_0}^0}^{\!\!\! *} \big\|
      \, \delta_x\big( {\ell_{\alpha_0}^1} \big)
       \\
       &
       \;\;
      + \delta_x\big( {\ell_{\alpha_0}^1} ^{\!\!\! *} \big)
      \, \big\|\ell_{\alpha_0}^0 \big\|
      +\big\| {\ell_{\alpha_0}^1}^{\!\!\! *} \big\|
      \, \delta_x\big( {\ell_{\alpha_0}^0} \big)
      \\
      & \;\;
      + \delta_x\big( {\ell_{\alpha_0}^1} ^{\!\!\! *} \big)
      \, \big\|\ell_{\alpha_0}^1 \big\|
      +\big\| {\ell_{\alpha_0}^1}^{\!\!\! *} \big\|
      \,\delta_x\big( {\ell_{\alpha_0}^1} \big)
     \end{split}
     &
    \begin{split}
      &\text{if } \alpha=\imath(\alpha_0) \in\imath(\mathcal{I}_0).
    \end{split}
    \\
    \\
          \delta_x(\ell_\alpha^*) \,  \|\ell_\alpha\|
     +\|\ell_\alpha^*\| \, \delta_x(\ell_\alpha^*)
      & \text{if $\alpha\not\in\imath(\mathcal{I}_0)$}
  \end{cases}
\end{split}
\end{equation*}
\end{lemma}

\begin{proof}
  The statement is a direct consequence of the bounds \eqref{bo1} and
  \eqref{bo2} below.

  For each $u\in\mathcal{A}^0$, $x\in\bb Z^d$,
  $h\in\{1,\dots,N\}$, and $f\in\mathcal{A}^1$,
  \begin{equation}
    \label{bo1}
    \big\| E_{x,h}[u,f] - [u, E_{x,h} f] \big\|
    \leq \sum_{y\in\bb Z^d}\gamma_{x,y}(u) \delta_y(f),
  \end{equation}
  where
  \begin{equation*}
    \gamma_{x,y}(u)
    =  2\eta (1+ \eta^2\delta_{x,y} ) \delta_x(u).
  \end{equation*}

  For each $u,v\in\mathcal{A}^0$, $x\in\bb Z^d$,
  $h\in\{1,\dots,N\}$, and $f\in\mathcal{A}^1$,
  \begin{equation}
    \label{bo2}
          \big\| E_{x,h} \big( u [f,v] + [ u,f] v \big)
    -u [E_{x,h} f, v ] - [u, E_{x,h} f] v \big\|
    \leq \sum_{y\in\bb Z^d}\gamma_{x,y}(u,v)\delta_y(f),
  \end{equation}
  where
  \begin{equation*}
    \gamma_{x,y}(u,v)
    =
   4\eta^3\big(1 +\eta^{-2} \delta_{x,y}\big) 
        \Big(\delta_x(u) \|v\| + \| u\|\delta_x(v)
        \Big).
  \end{equation*}

  To prove \eqref{bo1}, let $X=\{x_1,\dots,x_m\}\in\Pfin$ with $x_1=x$ and $u\in\mathcal{A}_X$.
  By \eqref{eq:spdec} and  \eqref{eq:Fx=}  
  \begin{equation}
    \label{eq:boh}
    \begin{split}
      & E_{x,h}[u,f]-[u, E_{x,h} f]=E_{x,h}[u,E_{x,0}f]\\
      & \qquad + \sum_{k=1}^N  \Big\{  E_{x,h} (u e_{x,k})
      (E_{x,k} f) - (E_{x,k}f) E_{x,h}(e_{x,k}u)\Big\}
      - \big[u,E_{x,h} f \big]
    \end{split}
  \end{equation}
  where we used that $E_{x,h} e_{x,k} =\delta_{h,k}$.
  By \eqref{eq:Fx} and recalling that $e_0=\id_{H}$, the first term
  on the right-hand side above can be expanded as
    \begin{equation*}
      E_{x,h}[u,E_{x,0}f] = \big[ E_{x_1,h}u , E_{x_1,0}f \big] 
      = \sum_{i=2}^m \Big[E_{x_1,h}u\,, \Big(
      \prod_{j< i}E_{x_j,0}\Big) F_{x_i} f
      \Big]
  \end{equation*}
  whose operator norm   is bounded above by
  \begin{equation*}
    2\|E_{x,h}u\|\sum_{y\in X\setminus\{x\}} \|F_{y}f \|
    \leq 2\eta\|E_{x,h}u\|\sum_{y\in X\setminus\{x\}}\sum_{k=1}^N \|E_{y,k}f \|.
  \end{equation*}
  On the other hand, by using again the identity \eqref{eq:spdec} for
  $u$, the second term on the right-hand side of \eqref{eq:boh} can be
  rewritten as
   \begin{equation}
     \label{eq:boh2}
     \big[E_{x,0}u,E_{x,h}f\big]  +
     \sum_{k,j=1}^N
     \Big\{
     C_{j,k}^h
     \big(E_{x,j}u \big)  \big( E_{x,k}f \big)
     -C_{k,j}^h
     \big( E_{x,k}f \big) \big(E_{x,j}u \big)
     \Big\}
   \end{equation}
   where $C_{j,k}^h :=E_{x,h} (e_{x,j}e_{x,k})$.
   Since $C^h_{j,k}=c^h_{j,k}\ind$ with $|c_{j,k}^h|\le \eta^3$, 
   the operator norm of the second term in \eqref{eq:boh2} is bounded above by
   \begin{equation*}
     2 \eta^3   \sum_{j=1}^N \|E_{x,j} u\| \,
     \sum_{k=1}^N \|E_{x,k} f\|.
   \end{equation*}
   Finally, again by the identity \eqref{eq:spdec} for $u$, the sum of
   the third term on the right-hand side of \eqref{eq:boh} and the
   first term of \eqref{eq:boh2} gives 
   \begin{equation*}
     -\sum_{k=1}^N \big[ (E_{x,k} u) e_{x,k} , E_{x,h} f \big] 
   \end{equation*}
   whose operator norm is bounded by
  \begin{equation*}
     2 \eta   \sum_{k=1}^N \|E_{x,k} u\| \, \|E_{x,h} f\|.
   \end{equation*}   
   Gathering the previous estimates we deduce \eqref{bo1}.

   To prove \eqref{bo2}, let $X\in\Pfin$ be such that $X\ni x$ and $u,v\in\mathcal{A}_X$.
   Enumerate as before the elements in $X\in\Pfin$ letting $x_1=x$ so that $X=\{x_1,\dots,x_m\}$.
   By the identity \eqref{eq:spdec},
   \begin{equation}
     \label{eq:Boh}
     \begin{split}
       & E_{x,h}\big( u [f,v] + [u,f] v \big)- u [E_{x,h} f ,v ]
       - [u,E_{x,h} f] v \big)
       \\
       &
       = E_{x,h}\big( u [E_{x,0} f,v] + [u, E_{x,0} f] v  
     +  2 u (F_{x} f ) v 
     -  u v (F_{x} f )  - (F_{x} f ) u v  
     \big)
     \\
     &\quad- 2 u (E_{x,h} f) v   + u v (E_{x,h} f)
     +(E_{x,h} f) uv. 
    \end{split}
  \end{equation}
  By \eqref{eq:Fx}
  \begin{equation}
    \label{bo2.1}
    \begin{split}
    &  E_{x,h}\big( u [E_{x,0} f,v] + [u, E_{x,0} f] v \big)
      = \sum_{i=2}^m 
      E_{x,h}\Big\{ 2u\Big( \Big(\prod_{j< i}E_{x_j,0}\Big)  F_{x_i} f \Big)v  \\
      &\qquad \qquad
      - uv \Big( \Big(\prod_{j< i}E_{x_j,0} \Big)  F_{x_i} f \Big)
      - \Big( \Big(\prod_{j< i}E_{x_j,0} \Big)  F_{x_i} f \Big) uv \Big\}.
     \end{split}
   \end{equation}
   We claim that for $g\in \mathcal{A}_{\{x\}^\mathrm{c}}$ and $u,v\in \mc A$
   \begin{equation}
     \label{dugv}
     \big\| E_{x,h} \big( u g v \big) \big\|
     \leq \eta^2 \|g\| \big(\delta_x(u) \|v\| + \|u\| \delta_x(v) \big).
   \end{equation}
   Indeed, by using the identity \eqref{eq:spdec} for $u$,
   \begin{equation*}
     \begin{split}
      E_{x,h} \big( u g v \big) & =
      (E_{x,0} u) g (E_{x,h}v)
      + \sum_{j=1}^N (E_{x,j} u) \, E_{x,h} ( e_{x,j} g v)
     \end{split}     
   \end{equation*}
   so that
   \begin{equation*}
     \big\|  E_{x,h} \big( u g v \big) \big\|
     \le \|g \|
     \Big( \|u\|  \|E_{x,h}v\|
     +  \eta^2 \sum_{j=1}^N
     \| E_{x,j} u \| \|v\| \Big)
   \end{equation*}
   which implies \eqref{dugv}. By using \eqref{dugv} we then bound the
   operator norm of the right-hand side of \eqref{bo2.1} by
   \begin{equation*}
     4 \eta^3\, 
     \sum_{j=1}^N \big(
     \|E_{x,j}u\| \|v\| + \|u\| \|E_{x,j}v \| \big)
     \, 
     \sum_{y\in X\setminus\{x\}}
     \sum_{k=1}^N \big\|E_{y,k} f \big\|.
   \end{equation*}

   By using \eqref{eq:spdec} and \eqref{eq:Fx} first for $u$ and then for
   $v$, noticing that $E_{x,h}F_x=E_{x,h}$,   
   \begin{equation}
     \label{bo2.1.1}
     \begin{split}
       & E_{x,h}\big(2 u (F_{x} f ) v -  u v (F_{x} f )  - (F_{x} f ) u v \big)
        =2 (E_{x,0}  u) (E_{x,h} f) (E_{x,0}  v)\\
       &\quad -  (E_{x,0}  u) (E_{x,0}  v) (E_{x,h} f) 
       - (E_{x,h} f) (E_{x,0}  u) (E_{x,0}  v) 
       \\
       &\quad
       + \sum_{k,j=1}^N
       \Big( 2 C^h_{k,j} (E_{x,0}  u) (E_{x,k} f)  (E_{x,j}  v)
       - C^h_{j,k} (E_{x,0}  u) (E_{x,j}  v) (E_{x,k} f)
       \\
       &\quad\qquad\qquad
       - C^h_{k,j} (E_{x,k} f)  (E_{x,0}  u) (E_{x,j}  v)  \Big)
       \\
       &\quad + 
       E_{x,h}\Big(
       2  (F_{x} u) (F_{x} f)  v
       - (F_{x} u)  v   (F_{x} f)
       - (F_{x} f)(F_{x} u) v   
     \Big),
   \end{split}
   \end{equation}
   where $C^h_{k,j}$ is defined below \eqref{eq:boh2}.
   Using that $\|C^h_{k,j}\|\le \eta^3$, the operator norm of the third and fourth lines on \eqref{bo2.1.1} is 
   bounded by
   \begin{equation*}
     4 \eta^3\, \|u\| 
     \sum_{j=1}^N  \|E_{x,j}v\|
     \sum_{k=1}^N  \|E_{x,k}f \|.
   \end{equation*}
   By using Lemma~\ref{lem:exh}, the operator norm of the fifth line on \eqref{bo2.1.1} can be 
   bounded by
     \begin{equation*}
     4 \eta^3\, \|v\| 
     \sum_{j=1}^N  \|E_{x,j}u\|
     \sum_{k=1}^N  \|E_{x,k}f \|.
   \end{equation*}
   Finally, again by the identity \eqref{eq:spdec} for $u$ and $v$,
   the sum of the last line on the right-hand side of \eqref{eq:Boh} and the
   first three terms on the right-hand side of \eqref{bo2.1.1} gives
   \begin{equation*}
     \begin{split}
       &
       -2 (E_{x,0} u) (E_{x,h}f) (F_{x} v)
       + (E_{x,0} u) (F_{x} v) (E_{x,h}f)
       + (E_{x,h}f) (E_{x,0} u) (F_{x} v)
       \\
       &- 2 (F_{x} u) (E_{x,h}f) v 
       + (F_{x} u) v (E_{x,h}f)
       + (E_{x,h}f) (F_{x} u) v
     \end{split}
   \end{equation*}
   whose operator norm is bounded by
   \begin{equation*}
     4 \eta\, 
     \sum_{j=1}^N \big( \|E_{x,j}u\| \|v\| +\|u\| \|E_{x,j}v\|  \big)
     \, \|E_{x,h}f \|.
   \end{equation*}
   Gathering the previous estimates we deduce \eqref{bo2}.
\end{proof}

Recalling $\theta_{x,y}$ has been defined below \eqref{a.6}, let $\Theta$ be the operator on $\ell_1(\bb Z^d)$ with kernel
$\big(\theta_{x,y}\big)_{x,y\in\bb Z^d}$.  Namely,
\begin{equation}\label{eqG}
  (\Theta \beta)_x
  = \sum_{y\in\bb Z^d} \theta_{x,y} \beta_y ,\qquad x\in\bb Z^d,\,\beta\in\ell_1(\bb Z^d).
\end{equation}
Recalling that $\lambda_1$ is the spectral gap of the unperturbed
single site Lindblad generator, the following lemma provides an a priori bound on the resolvent equation.

\begin{lemma}\label{lem:Gamma}
  Assume \eqref{eq:C0} and $M<+\infty$. 
  The operator $\Theta$ on $\ell_1(\bb Z^d)$ satisfies
  $\|\Theta \|_{\ell_1\to\ell_1}\leq M$.
  Furthermore, if $f,g\in\mc A^1$ satisfy
  \begin{equation}\label{eqRis}
    (\lambda -{\mc L})g=f
  \end{equation}
  for some $\lambda> 0$ such that $\lambda+\lambda_1>M$, then
  \begin{equation}\label{eqBound1}
    \delta(g)
    \leq
    \big(\lambda+\lambda_1-\Theta\big)^{-1}\delta(f)
   \qquad\text{pointwise.}
  \end{equation}
   In particular,
  \begin{equation}\label{eqBound2}
    \vertiii{g}\leq (\lambda+\lambda_1-M)^{-1}\vertiii{f}
    \,.
  \end{equation}
\end{lemma}
\begin{proof}
  Definition \eqref{a.6} readily implies the bound
  $\|\Theta\|_{\ell_1\to\ell_1}\leq M$.  Assuming \eqref{eqRis}, for
  $x\in\bb Z^d$ and $h\in\{1,\dots,N\}$ we have
  \begin{align*}
     \lambda E_{x,h}g
     & =
      E_{x,h}f+E_{x,h}\mc L g \\
    & =
      E_{x,h} f+\mc LE_{x,h}g
      +
      E_{x,h}\mc L_0g-\mc L_0E_{x,h} g
            +
      E_{x,h}\mc L_1g-\mc L_1E_{x,h} g.
  \end{align*}
  Applying Lemma \ref{lem24}, we thus get
  \begin{equation}
    \label{a.10}
    (\lambda+\lambda_h)E_{x,h}g
    =
    E_{x,h} f+\mc LE_{x,h}g
    +
    E_{x,h}\mc L_1g-\mc L_1E_{x,h} g
    \,.
  \end{equation}
  Since, as proven in Lemma~\ref{lem:accr}, $-\mc L$ is accretive, for
  each $x\in\bb Z^d$ and $h=1,\ldots,N$, there exists
  $\wp\in\mc A^\prime$ such that $\|\wp\|_{\mc A^\prime}=1$,
  $\wp(E_{x,h} g)=\|E_{x,h} g\|$, and $\re\wp(\mc LE_{x,h} g)\leq0$.
  Pairing both sides of \eqref{a.10} with $\wp$ and taking real
  parts we deduce
  \begin{equation}
    \label{eq:lmabdapiu}
    \begin{split}
          (\lambda+\lambda_h)\|E_{x,h}g\| 
    & \leq
      \re\wp(E_{x,h} f)
      +
      \re\wp(E_{x,h}\mc L_1g-\mc L_1E_{x,h} g) \\
    & \leq
      \|E_{x,h}f\|
      +
      \|E_{x,h}\mc L_1g-\mc L_1E_{x,h} g\|.
    \end{split}
  \end{equation}

  Lemma~\ref{lem246} and the definition of $\theta_{x,y}$ below \eqref{a.6} imply that for each $h\in\{1,\dots,N\}$
  \begin{equation*}
    \|E_{x,h}\mathcal{L}_1g-\mathcal{L}_1 E_{x,h}g\|
    \leq\frac{1}{N}\sum_{y\in\bb Z^d}\theta_{x,y}\delta_y(g).
  \end{equation*}
  Since $\lambda_h\ge \lambda_1$,
  summing over  $h\in\{1,\dots,N\}$ the bound \eqref{eq:lmabdapiu} we get
  \begin{equation*}
    (\lambda+\lambda_1)\delta(g) \leq \delta(f) + \Theta \delta(g) \qquad\text{pointwise.}
  \end{equation*}
  Equivalently,
  \begin{equation*}
    \delta(g) \leq \frac1{\lambda+\lambda_1}\delta(f) +
\frac1{\lambda+\lambda_1}\Theta \delta(g) \qquad\text{pointwise.}
\end{equation*}
Since $\Theta$ is positive operator and
$\delta(f),\delta(g)\in\ell_1(\bb Z^d)$, this inequality can be iterated to
obtain
\begin{equation*}
\delta(g) \leq \sum_{k=0}^{n-1}
\frac{\Theta^k}{(\lambda+\lambda_1)^{k+1}}\delta(f) +
\frac{\Theta^n}{(\lambda+\lambda_1)^{n}}\delta(g) \qquad\text{pointwise.}  
\end{equation*}
By assumption $\lambda+\lambda_1>M$, hence the last term in the
right-hand side above vanishes as $n\to\infty$.  The bound
\eqref{eqBound1} follows.
Since
$\|(\lambda+\lambda_1-\Theta)^{-1}\|_{\ell_1\to\ell_1}\leq(\lambda+\lambda_1-M)^{-1}$, 
the bound \eqref{eqBound2} is obtained by taking the $\ell_1(\bb Z^d)$-norm of
both sides in \eqref{eqBound1}.
\end{proof}

\begin{proof}[Proof of Theorem \ref{thm1}]
  By Lemma \ref{lem:accr} $-\mc L$
  is accretive thus closable, see e.g.\ \cite[Ex.~X.52]{RS2}.
  The proof that its closure $\bar{\mc L}$ generates a QFS on $\mc A$
  is achieved by the following steps.

  \medskip

  \noindent\emph{Step 1.}
  The operator $\bar{\mc L}$ generates a strongly continuous
  contraction semigroup $(\mc P_t)_{t\geq0}$ on the Banach space $\mc A$.

  In view of the accretivity of $-\mc L$ and the Lumer-Phillips theorem, see
  e.g.\ \cite[Thm.~X.48]{RS2}, it is enough to show
  $\Ran(\lambda-\bar{\mc L})=\mc A$ for some $\lambda>0$, where $\Ran(T)$ denotes the image of the linear operator $T$.
  To this end it suffices to show that $\Ran(\lambda-\mc L)$ is dense in $\mc A$.
  Pick a sequence $\Lambda_n\in\Pfin$, $n\in\bb N$, such that
  $\Lambda_n\subset\Lambda_{n+1}$ and $\bigcup_n\Lambda_n=\bb Z^d$.
  Let also $\mc L^{(n)}$ be the finite volume generator defined in \eqref{eq:Lfin} with $\Lambda$ replaced by $\Lambda_n$
  and observe that $-\mc L^{(n)}$ is an accretive bounded operator on $\mc A$.
  Hence, $\mc L^{(n)}$ generates a strongly continuous contraction semigroup on $\mc A$
  denoted by $\big(\mathcal{P}_t^{(n)}\big)_{t\geq0}$.  By the
  Lumer-Phillips theorem we then deduce
  $\Ran(\lambda-\mc L^{(n)})=\mc A$ for any $n\geq1$ and $\lambda>0$.
  Fix $f\in\mc A^1$, choose $\lambda>0$ such that
  $\lambda+\lambda_1>M$, and set $g_n=(\lambda-\mc L^{(n)})^{-1}f$,
  $f_n=(\lambda-\mc L)g_n$.  We claim that
  \begin{equation}\label{eq:limitefn}
    \|f_n-f\|\to0 \quad\text{ as }\quad n\to\infty
    \,.
  \end{equation}
  Since $\mc A^1$ is dense in $\mc A$, this yields the density of
  $\Ran(\lambda-\mc L)$ in $\mc A$.

  To prove \eqref{eq:limitefn} we decompose $\mc L^{(n)}=\mc L_0^{(n)}+\mc L_1^{(n)}$ where $\mc L_0^{(n)}=\sum_{x\in\Lambda_n}L^0_x$ and $\mc L_1^{(n)}=\sum_{\alpha\colon\chi(\alpha)\subset\Lambda_n}L^1_\alpha$.
  Set $\theta_{x,y}^{(n)}=\theta_{x,y}$ if $x,y\in\Lambda_n$ and $\theta_{x,y}^{(n)}=0$ otherwise and let $\Theta^{(n)}$ be the operator on $\ell_1(\bb Z^d)$ with kernel $\theta_{x,y}^{(n)}$.
  We the claim that
  \begin{equation*}
    \delta(g_n)
    \leq \big(\lambda+\lambda_1-\Theta^{(n)}\big)^{-1}\delta(f_n),
    \qquad\text{pointwise.}
  \end{equation*}
  This follows indeed from Lemma~\ref{lem246} and the argument in the proof of Lemma~\ref{lem:Gamma} to deduce \eqref{eqBound1}.
  Since $\theta_{x,y}^{(n)}\leq \theta_{x,y}$, the previous bound implies
  \begin{equation}\label{eq:stima4}
    \delta(g_n)
    \leq
     (\lambda+\lambda_1-\Theta)^{-1}\delta(f),
    \qquad\text{pointwise.}
  \end{equation}
  The argument in proof of Lemma~\ref{lem:LX} and more precisely the bound \eqref{eq:kbound} implies 
  \begin{equation*}
    \|f_n-f\| = \big\| (\mc L^{(n)}-\mc L)g_n \big\| \leq
    \sum_{x\in \bb Z^d} C_0^{(n)}(x)\delta_x(g_n)
  \end{equation*}
  where
  \begin{equation*}
    C_0^{(n)}(x) = 2\eta\sum_{\substack{\alpha\in \chi^{-1}(\{x\})\\ \chi(\alpha)\cap \Lambda_n \neq\emptyset}} (\|k_\alpha\|+2\|\ell_{\alpha}\|^2).
  \end{equation*}
  Recalling \eqref{eq:C0} we now observe that $C_0^{(n)}(x) \le C_0$ and $C_0^{(n)}$ vanishes pointwise as $n\to \infty$.
  Hence, by the bound \eqref{eq:stima4} and dominated convergence we conclude the proof of \eqref{eq:limitefn}.

\medskip

\noindent\emph{Step 2.}
$\mc P_t(\mc A^1)\subset\mc A^1$ and
$\vertiii{\mc P_t f}\leq e^{(M-\lambda_1)t}\vertiii{f}$,  $t\geq0$ and $f\in\mc A^1$.

We first observe that, by the relationship between the semigroup
$\mc P_t$ and the resolvent $(\lambda-\bar{\mc L})^{-1}$ provided by the Hille-Yoshida theorem, for each
$u\in\mc A$ and $t>0$ we have
\begin{equation}\label{eq:lorenzo}
  \mc P_t u
  =
  \lim_{n\to\infty}
  \Big(\frac nt\Big)^n
  \Big(\frac nt-\bar{\mc L}\Big)^{-n} u.
\end{equation}
For $f$ as in the statement, let $g_n$ and $f_n$ be defined as in Step
1, where we recall that $\lambda+\lambda_1>M$.  Since, as proven in
Step 1, $f_n\to f$, and $(\lambda-\mc L)^{-1}$ is bounded,
$g_n=(\lambda-\mc L)^{-1}f_n$ is a Cauchy sequence, whose limit is
denoted by $g\in\mc A$.  We deduce that the sequence
$\mc L g_n=\lambda g_n-f_n$ has a limit.  Hence, $f$ belongs to the
domain of $\bar{\mc L}$ and $f=(\lambda-\bar{\mc L})g$.  By taking the
limit $n\to\infty$ in \eqref{eq:stima4}, we deduce that
$\delta(g)\leq(\lambda+\lambda_1-\Theta)\delta(f)$ pointwise, in
particular $g\in\mc A^1$.  Thus
$(\lambda-\bar{\mc L})^{-1}\mc A^1\subset\mc A^1$, which, by
\eqref{eq:lorenzo}, proves that $\mc P_t\mc A^1\subset\mc A^1$, $t\geq0$.

As just proven, if $\lambda+\lambda_1>M$,
$$
\delta((\lambda-\bar{\mc L})^{-1}f)\leq
(\lambda+\lambda_1-\Theta)^{-1}\delta(f)\qquad\text{pointwise.}
$$
Since $\Theta$ is a positive operator, by iterating this bound we
deduce that for every $n\in\bb N$
$$
\delta((\lambda-\bar{\mc L})^{-n}f)\leq
(\lambda+\lambda_1-\Theta)^{-n}\delta(f)\qquad\text{pointwise}.
$$
Hence, by \eqref{eq:lorenzo}, for $t>0$
\begin{align*}
  & \vertiii{\mc P_t f}
    =
    \sum_{x\in\bb Z^d}
    \delta_x(\mc P_t f)
    =
    \sum_{x\in\bb Z^d}
    \lim_{n\to\infty}
    \Big(\frac nt\Big)^n
    \delta_x\Big(\Big(\frac nt-\bar{\mc L}\Big)^{-n} f\Big) \\
  & \leq
    \sum_{x\in\bb Z^d}
    \lim_{n\to\infty}
    \Big(\frac nt\Big)^n
    \Big(\Big(\frac nt+\lambda_1-\Theta\Big)^{-n}
    \delta(f)\Big)_x \\
  & =
    \sum_{x\in\bb Z^d}
    \big(e^{(\Theta-\lambda_1)t}
    \delta(f)
    \big)_x
    \leq
    e^{(M-\lambda_1) t}\vertiii{ f}
    \,,
\end{align*}
where we used $\|\Theta\|_{\ell_1\to\ell_1}\leq M$ in the last
inequality.

\medskip

\noindent\emph{Step 3.}
For each $f\in\mc A$ and $t\geq0$ the sequence $\mc P_t^{(n)}f$ converges to $\mc P_tf$ as $n\to\infty$.

Since $\mc P_t$ and $\mc P^{(n)}_t$ are contractions on $\mc A$, and
$\mc A^1$ is dense in $\mc A$, it is enough to show the statement for
each $f\in\mc A^1$.  By Step 2 and standard interpolation, we have
$$
\mc P_tf-\mc P_t^{(n)}f = \int_0^tds\frac{d}{ds}\big(\mc
P_{t-s}^{(n)}\mc P_sf\big) = \int_0^tds\,\mc P_{t-s}^{(n)} (-\mc
L^{(n)}+\mc L) \mc P_sf \,,
$$
so that
$$
\|\mc P_tf-\mc P_t^{(n)}f\| \leq \int_0^tds \| (-\mc L^{(n)}+\mc L) \mc
P_sf\|\,.
$$
By the argument below \eqref{eq:stima4}, dominated convergence,
and again Step 2, we then conclude that the right-hand side vanishes
as $n\to\infty$.

\medskip

\noindent\emph{Conclusion.}
By Step 1, $\bar{\mc L}$ generates a strongly continuous contraction
semigroup $(\mc P_t)_{t\geq0}$ on the Banach space $\mc A$.  To show
that $(\mc P_t)_{t\geq0}$ is a QFS we need to prove that $\mc P_t$ is
a completely positive operator on the $C^*$-algebra $\mc A$ for each $t\geq0$.
Observing that $\mc L^{(n)}$ is a finite rank operator, the standard theory of QFS on finite
dimensional $C^*$-algebra \cite{Lindblad76} implies that $\mc P_t^{(n)}$ is a completely positive operator on $\mc A$ for each
$t\geq0$ and $n\in\bb N$.
As follows from Step 3, for each $t\geq0$, the operator $\mc P_t$ is
the strong limit of $\mc P_t^{(n)}$, and therefore also completely
positive.

In view of Step 2 and \cite[Thm.~X.49]{RS2} $\mc A^1$ is a core for
$\bar{\mc L}$.  Claims (iii) and (iv) are the content of Steps 3 and 2, respectively.

Finally, to prove item (v), pick $\mu\in\mathcal{S}$ and set
\begin{equation*}
  \mu^T
  :=\frac{1}{T}\int_0^T\mu\mathcal{P}_t\,dt,
  \qquad
  T\in(0,\infty).
\end{equation*}
Since $\|\mu^T\|_{\mathcal{A}'}=1$, the Banach-Alaoglu theorem yields the existence of $\pi\in\mathcal{A}'$ with $\|\pi\|_{\mathcal{A}'}\leq1$ and a  sequence $T_n\to\infty$ such that $\mu^{T_n}\to\pi$ weakly* in $\mathcal{A}'$.
Moreover, $\pi$ is positive and, since $\mathcal{A}$ is unital, $\pi(\ind)=1$  so that $\pi\in\mathcal{S}$.
Fix $s\geq0$, by taking the limit $n\to\infty$ in
\begin{equation*}
  \mu^{T_n}\mathcal{P}_s
  =\frac{1}{T_n}\int_0^{T_n}\mu\mathcal{P}_{t+s}\,dt
  =\mu^{T_n}-\frac{1}{T_n}\int_0^{s}\mu\mathcal{P}_{t}\,dt+\frac{1}{T_n}\int_{T_n}^{T_n+s}\mu\mathcal{P}_{t}\,dt,
\end{equation*}
we deduce that $\pi=\pi\mathcal{P}_s$ hence $\pi$  is stationary.

To prove the infinitesimal characterization of stationary state, we first observe that, in view of Lemma~\ref{lem:LX} and items (i)--(ii), the stated condition is in fact equivalent to $\pi(\bar{\mc L}f)=0$ for any $f$ in the domain of $\bar{{\mc L}}$.
Clearly, the invariance of $\pi$ implies this condition.
To show the converse, from the infinitesimal invariance we deduce that $\pi\big(\lambda(\lambda-\mathcal{L})^{-1}f\big)=\pi(f)$ for any $\lambda>0$ and $f\in\mathcal{A}$.
Therefore, by iterating, $\pi\big(\lambda^n(\lambda-\mathcal{L})^{-n}f\big)=\pi(f)$.
We conclude by using \eqref{eq:lorenzo}.
\end{proof}

\begin{remark}\label{rem:ecellenza}
  The proof of Theorem \ref{thm1} given above actually implies that
  for each $f\in\mc A^1$ and $t\geq0$
  \begin{equation*}
    \delta(\mathcal{P}_tf) \leq e^{-\lambda_1 t} e^{t\Theta}\delta(f) \qquad
\text{pointwise.}
  \end{equation*}
\end{remark}

\subsection{Perturbative criterion for ergodicity}
\label{sub:2}

Before discussing the proof of Theorem~\ref{thm2}, we show that the Lieb-Robinson bounds hold for the QFS $\big(\mathcal{P}_t\big)_{t\geq0}$: the evolution of the product of two observables, with distant ``support'', can be approximated by the product of their respective evolution on a suitable  time scale.

\begin{proposition}\label{pro:dec}
Set
\begin{equation}\label{eq:omega}
\omega_{x,y}:=
8 \eta^2\sum_{\substack{\alpha\in \mathcal{I}\\ \chi(\alpha)\supset\{x,y\}}}\|\ell_{\alpha}\|^2,\qquad x,y\in\bb Z^d
\end{equation}
and  assume there exists $\xi>0$ such that
\begin{align}
M_\xi &:= 
\sup_{y\in\bb Z^d}\sum_{x\in\bb Z^d}
\theta_{x,y}e^{\xi|x-y|}<\infty,
\label{eq:Meta}\\
\Omega_\xi &:=
\sup_{x,y\in\bb Z^d}\omega_{x,y}e^{\xi|x-y|}<\infty
.\label{eq:omegaeta}
\end{align}
Then, for each $\Lambda_1,\Lambda_2\in\Pfin$, 
$f_1\in\mc A_{\Lambda_1}$, $f_2\in\mc A_{\Lambda_2}$,
and $t\geq0$,
$$
\|\mc P_t(f_1f_2)-(\mc P_tf_1)(\mc P_tf_2)\|
\leq
\Omega_\xi
\frac{e^{2(M_\xi-\lambda_1)t}-1}{2(M_\xi-\lambda_1)}
e^{-\xi\,\mathrm{dist}(\Lambda_1,\Lambda_2)}
\vertiii{ f_1} \vertiii{ f_2}
.
$$
\end{proposition}
\begin{lemma}\label{lem:dec}
For any $f_1,f_2\in\mc A^1$ and $t\geq0$,
$$
\|\mc P_t(f_1f_2)-(\mc P_tf_1)(\mc P_tf_2)\|
\leq
\sum_{x,y\in\bb Z^d}
\omega_{x,y}\int_0^t ds
e^{-2\lambda_1 s}
(e^{s\Theta}\delta(f_1))_x
(e^{s\Theta}\delta(f_2))_y
.
$$
\end{lemma}

\begin{proof}
Set
$F_t=\mc P_t(f_1f_2)-(\mc P_tf_1)(\mc P_tf_2)$.
By direct computation,
$$
\frac{d}{dt} F_t
=
\mc LF_t+G_t
$$
where
$$
G_t
:=
\mc L((\mc P_tf_1)(\mc P_tf_2))
-(\mc L\mc P_tf_1)(\mc P_tf_2) -(\mc P_tf_1)(\mc L\mc P_tf_2)
\,.
$$
Since $F_0=0$ and $\mc P_t$ is a contraction on $\mathcal{A}$
we deduce
$$
\|F_t\|
=
\Big\|\int_0^t\mc P_{t-s}G_s\,ds\Big\|
\leq
\int_0^t\|G_s\|ds
\,.
$$
Given $g_1,g_2\in\mc A^1$, by direct computation,
\begin{equation*}
 \mc L(g_1g_2)
-(\mc Lg_1)g_2
-g_1(\mc Lg_2) 
=
2\sum_{\alpha\in\mc I}[\ell_{\alpha}^*,g_1][g_2,\ell_{\alpha}]
\end{equation*}
whose operator norm, using \eqref{eq:kbound}, is bounded by $\sum_{x,y}\omega_{x,y}\delta_x(g_1)\delta_y(g_2)$.
In view of Remark \ref{rem:ecellenza}, the statement follows by choosing $g_i=\mc P_tf_i$, $i=1,2$.
\end{proof}

\begin{proof}[{Proof of Proposition \ref{pro:dec}}]
Recalling \eqref{eqG}, assumption \eqref{eq:Meta} readily implies, for $s\geq0$,
$$
\sup_{y\in\bb Z^d}\sum_{x\in\bb Z^d}
\big(e^{s\Theta}\big)_{x,y}e^{\xi|x-y|}
\leq e^{sM_\xi} .
$$
This bound, together with assumption \eqref{eq:omegaeta}, yields
$$
\sum_{x',y'\in\bb Z^d}\omega_{x',y'}\big(e^{s\Theta}\big)_{x',x}\big(e^{s\Theta}\big)_{y',y}
\leq \Omega_\xi e^{2 M_\xi s-\xi |x-y|}.
$$
The statement now follows from Lemma \ref{lem:dec} and elementary computations.
\end{proof}

\begin{proof}[Proof of Theorem \ref{thm2}]
To prove (i), we observe that for $f\in\mc A^1$ the sequence $(\mc P_tf)_{t\geq0}$  is Cauchy in $\mc A$ as $t\to\infty$.
Indeed, for $t\geq s\geq 0$ 
\begin{equation*}
  \mc P_tf-\mc P_sf
=
\int_s^t\mc L\mc P_rf
\,dr,
\end{equation*}
so that
\begin{equation}\label{eq:t-s}
\begin{split}
& \|\mc P_tf-\mc P_sf\|
\leq
\int_s^t\|{\mc L}\mc P_rf\|\,dr
\leq
C_0\int_s^t\vertiii{\mc P_rf}\,dr \\
& \leq
C_0\int_s^te^{-(\lambda_1-M)r}\vertiii{ f}\,dr
=
\frac{C_0}{\lambda_1-M} \big( e^{-(\lambda_1-M)s}-e^{-(\lambda_1-M)t} \big) \vertiii{ f}
\,,
\end{split}
\end{equation}
where we used Lemma~\ref{lem:LX} for the second inequality
and Theorem \ref{thm1}(iv) for the third one.
Let then $f_\infty=\lim_{t\to\infty}\mc P_t f\,\in\mc A$.
We claim that $f_\infty\in\bb C\ind$.
Indeed, by Lemma~\ref{lem:exh} and Theorem~\ref{thm1}(iv)
$E_{x,h} f_\infty=0$, $x\in\bb Z^d$, $h\in\{1,\dots,N\}$ which, as observed below \eqref{3norm}, yields $f_\infty\in\bb C\ind$.

Let $\pi\colon\mc A^1\to\bb C$ be defined by $f_\infty=\pi(f)\ind$.
We claim that $\pi$ extends to a state $\pi$ on $\mc A$.
The map $\pi$ is obviously linear and, since $\mc P_t$ is a contraction,
$|\pi(f)|=\lim_{t\to\infty}\|\mc P_t f\|\leq\|f\|$.
Hence, $\|\pi\|_{\mc A'}\leq 1$.
This bound implies that the map $\pi$ can be extended to a continuous 
linear functional on $\mc A$.
Furthermore $\pi(\ind)=1$.
It remains to show that $\pi(f)\geq0$ for $f\geq0$.
Since $(\mc P_t)_{t\geq0}$ is a QFS, we have that $\mc P_t f\geq0$,
and passing to the limit $t\to\infty$ we obtain $f_\infty\geq0$,
which implies $\pi(f)\geq0$.
The semigroup property and the definition of $\pi$ imply
that $\pi(\mc P_t f)=\pi(f)$ for any $t\geq0$ and $f\in\mc A^1$, 
hence $\pi$ is a stationary state.

To show uniqueness, if $\nu$ is a stationary state, for $f\in\mc A^1$ we have
$\nu(f)=\nu(\mc P_tf)$ which in the limit $t\to+\infty$ yields $\nu(f)=\nu(\pi(f)\ind)=\pi(f)$.
Hence $\nu=\pi$ since $\mc A^1$ is dense in $\mc A$.

To prove (ii) it suffices to take the limit $t\to\infty$ in \eqref{eq:t-s}.

Finally, we prove claim (iii). By the triangle inequality, for each $t\geq0$,
\begin{align*}
& |\pi(f_1f_2)-\pi(f_1)\pi(f_2)|
\leq
\|\mc P_t(f_1f_2)-\pi(f_1f_2)\ind\|
+
\|\mc P_tf_1-\pi(f_1)\ind\|\|\mc P_tf_2\|
\\
&\qquad
+ |\pi(f_1)| \|\mc P_tf_2-\pi(f_2)\ind\| 
+
\|\mc P_t(f_1f_2)-(\mc P_tf_1)(\mc P_tf_2)\| \\
&
\leq
Ce^{-(\lambda_1-M)t}\big(\vertiii{ f_1f_2}+\vertiii{ f_1} \|f_2\|+\|f_1\|\vertiii{ f_2}\big)
+\|\mc P_t(f_1f_2)-(\mc P_tf_1)(\mc P_tf_2)\|
\,,
\end{align*}
where we have used (ii) and have set $C=C_0/(\lambda_1-M)$.
The bound \eqref{dugv} implies
$$
\vertiii{f_1f_2}
\leq
N\eta^2\big(\vertiii{ f_1} \|f_2\|+\|f_1\|\vertiii{ f_2}\big)
\leq
N\eta^2\big(\|f_1\|+\vertiii{ f_1}\big)\big(\|f_2\|+\vertiii{ f_2}\big).
$$
Note that conditions \eqref{eq:Meta} and \eqref{eq:omegaeta} trivially hold for any $\xi>0$ 
by the finite range assumptions.
Moreover, denoting by $R$ the range,
we have 
$M_\xi\leq Me^{R\xi}$.
We can thus apply Proposition \ref{pro:dec} which yields
\begin{align*}
& |\pi(f_1f_2)-\pi(f_1)\pi(f_2)| \\
& \leq
\frac{C_\xi}2\Big(
e^{-(\lambda_1-M)t}
+
e^{2Me^{R\xi} t-\xi\,\mathrm{dist}(\Lambda_1,\Lambda_2)}
\Big)
\big(\|f_1\|+\vertiii{ f_1}\big)\big(\|f_2\|+\vertiii{ f_2}\big),
\end{align*}
where
$C_\xi=2N\eta^2\max\big\{C,\Omega_\xi/(2Me^{R\xi})\big\}$.

By choosing $t=\xi\big(\lambda_1-M+2Me^{R\xi})^{-1}\mathrm{dist}(\Lambda_1,\Lambda_2)$, we deduce (iii) with $C=C_\xi$ and $\zeta=(\lambda_1-M)\xi\big(\lambda_1-M+2Me^{R\xi}\big)^{-1}$.
\end{proof}

\subsection{Convergence of the specific quantum one-Wasserstein distance}
\label{sub:3}
We first recall the dual formulation of the quantum one-Wasserstein distance  $W_\Lambda$ in terms of a Lipschitz seminorm proven in  \cite[Prop. 8]{DD0}.
Given $\Lambda\in\Pfin$ we introduce the Lipschitz seminorm $\vertiii{\cdot}_{\Lambda,\mathrm{Lip}}$ on $\mathcal{A}_\Lambda$ as
$\vertiii{f}_{\Lambda,\mathrm{Lip}}:=\sup_{x\in\Lambda}\theta_x(f)$ where $  \theta_x(f):=2\inf_{g\in\mathcal{A}_{\Lambda\setminus\{x\}}}\|f-g\|$.
Then the quantum one-Wasserstein distance on the set of states on $\mathcal{A}_\Lambda$ can be represented as 
\begin{equation}\label{W=}
  W_{\Lambda}(\mu,\nu)
  =\sup_{\vertiii{f}_{\Lambda,\mathrm{Lip}}\leq 1}|\mu(f)-\nu(f)|.
\end{equation}

\begin{lemma}\label{appena}
  For each $\Lambda\in\Pfin$ and $f\in\mathcal{A}_\Lambda$
  \begin{equation*}
    \vertiii{f}
    \leq
    \frac{1}{2}N\eta|\Lambda|\,\vertiii{f}_{\Lambda,\mathrm{Lip}}.
  \end{equation*}
\end{lemma}

\begin{proof}
  Since $E_{x,h}f=E_{x,h}(f-g)$ for any $h\in\{1,\dots,N\}$ and $g\in\mathcal{A}_{\Lambda\setminus\{x\}}$, by Lemma~\ref{lem:exh}, we deduce $\|E_{x,h}f\|\leq\eta\theta_x(f)/2$.
  The statement follows.
\end{proof}

\begin{proof}[Proof of Theorem~\ref{thm3}]
  Since $\mathcal{L}$ is translation covariant, the QFS $(\mathcal{P}_t)_{t\geq0}$ has a translation invariant stationary state $\pi$, which is unique by Theorem~\ref{thm2}(i).
  Recalling \eqref{w=}, for each $\mu\in\mathcal{S}_\tau$ and $t\geq0$,
  \begin{align*}
    w(\mu\mathcal{P}_t,\pi)
    &=\sup_{\Lambda\in\Pfin}\frac{1}{|\Lambda|}\sup_{\vertiii{f}_{\Lambda,\mathrm{Lip}}\leq1}\big|\mu\big(\mathcal{P}_tf-\pi(f)\ind\big)\big|\\
    &\leq \sup_{\Lambda\in\Pfin}\frac{1}{|\Lambda|}\sup_{\vertiii{f}_{\Lambda,\mathrm{Lip}}\leq1}\big\|\mathcal{P}_tf-\pi(f)\ind\big\|.
  \end{align*}
  The stated bound, with $C=C_0(\lambda_1-M)^{-1}N\eta/2$, now follows directly from Theorem~\ref{thm2}(ii) and Lemma~\ref{appena}.
\end{proof}

\section{
Interacting qudits: examples
} 
\label{s:2-d}

To exemplify the abstract theory developed before, we next introduce simple dissipative quantum lattice systems and discuss when the perturbative criterion for ergodicity in Theorem~\ref{thm2} can be applied.

\subsection{Quantum spin systems: application of the perturbative criterion}
\label{sec:qss}

We consider a class of QFS with purely dissipative Lindblad generators $\mathcal{L}$  and show how the decomposition $\mathcal{L}=\mathcal{L}_0+\mathcal{L}_1$ can be achieved.
We focus on the case of translation covariant interactions with finite range.

Let $H=\bb C^2$, $A=\mathcal{B}(\bb C^2)$ and denote by $\mathcal{H}$ and $\mathcal{A}$ the Hilbert space and the $C^*$-algebra constructed in Section~\ref{sec:udq}.
Let $\mathcal{I}=\bb Z^d\times\{1,2,3\}$ and $\chi\colon\mathcal{I}\to \Pfin$ be the map  $(x,j)\mapsto\{y\colon |x-y|\leq R\}=:B_R(x)$.
Consider the informal Lindblad generator
\begin{equation}
  \label{eq:qss}
  \mathcal{L}
  =\sum_{\alpha\in\mathcal{I}}\big(\ell_\alpha^*[\,\cdot\,,\ell_\alpha]+[\ell_\alpha^*,\,\cdot\,]\ell_\alpha^*\big),
\end{equation}
where the jump operators $\ell_\alpha$ satisfy $\ell_{(x,j)}\in\mathcal{A}_{B_R(x)}$ and are translation covariant. 
As we will next show, under suitable conditions on these operators, Theorems \ref{thm1} and \ref{thm2} can be applied to deduce that the graph norm closure of $\mathcal{L}$ generates an ergodic QFS.

We denote by $\sigma_j$, $j=1,2,3$ the Pauli matrices
\begin{equation*}
  \sigma_1
  =
  \begin{pmatrix}
    0& 1\\
    1& 0
  \end{pmatrix},
  \quad
    \sigma_2
  =
  \begin{pmatrix}
    0& -i\\
    i& 0
  \end{pmatrix},
  \quad
    \sigma_3
  =
  \begin{pmatrix}
    1& 0\\
    0& -1
  \end{pmatrix}
\end{equation*}
and by $\sigma_{x,j}$, $j=1,2,3$, the corresponding elements of $\mathcal{A}_{\{x\}}$.

Fix $x_0\in\bb Z^d$, let $\mathcal{D}_{\{x_0\}}\subset\mathcal{A}_{\{x_0\}}$ be the real linear span of $\{\ind,\sigma_{x_0,3}\}$ and set
\begin{equation}\label{eq:inf}
  \begin{split}
      \Delta_{1,2}&=\min\big\{\|\sigma_{x_0,1}g-\ell_{(x_0,1)}\|+\|\sigma_{x_0,2}g-\ell_{(x_0,2)}\|\colon g\in\mathcal{D}_{\{x_0\}}\big\},\\
  \Delta_3&=\min\big\{\|g-\ell_{(x_0,3)}\|\colon g\in\mathcal{D}_{\{x_0\}}\big\}.
  \end{split}
\end{equation}
Next, choose
\begin{align*}
  a&\in\mathrm{argmin}\big\{\|\sigma_{x_0,1}g-\ell_{(x_0,1)}\|+\|\sigma_{x_0,2}g-\ell_{(x_0,2)}\|\colon g\in\mathcal{D}_{\{x_0\}}\big\},\\
  b&\in\mathrm{argmin}\big\{\|g-\ell_{(x_0,3)}\|\colon g\in\mathcal{D}_{\{x_0\}}\big\}.
\end{align*}
By translation covariance, $\Delta_{1,2}$, $\Delta_{3}$ do not depend on $x_0$.
Furthermore, identifying $\mathcal{A}_{\{x_0\}}$ with $A=\mathcal{B}(\bb C^2)$, there exist reals $\alpha_0,\alpha_1,\beta_0,\beta_1$ such that
\begin{equation*}
  a=
  \begin{pmatrix}
    \alpha_0 & 0 \\
    0 & \alpha_1
  \end{pmatrix},
  \qquad
    b=
  \begin{pmatrix}
    \beta_0 & 0 \\
    0 & \beta_1
  \end{pmatrix}.
\end{equation*}
Hereafter, we assume $\alpha_0^2+\alpha_1^2>0$.

Introduce on  $A$ the Lindblad generator
\begin{equation}\label{eq:Lin1}
  L_0
  =\sum_{j=1}^2\big(a \sigma_j [\,\cdot\,,\sigma_ja]+[a \sigma_j, \,\cdot\,]\sigma_ja\big)+
  b\sigma_3 [\,\cdot\,,\sigma_3b]+[b \sigma_3, \,\cdot\,]\sigma_3b.
\end{equation}
By direct computation, the Lindblad generator $L_0$ is self-adjoint with respect to the GNS inner product induced by the density matrix
\begin{equation*}
  \rho
  =\frac{1}{\alpha_0^2+\alpha_1^2}
  \begin{pmatrix}
    \alpha_0^2 & 0 \\
    0 & \alpha_1^2
  \end{pmatrix}.
\end{equation*}
The eigenvalues of $-L_0$ are $\lambda_0=0$, $\lambda=4(\alpha_0^2+\alpha_1^2)$, and $\mu=2(\alpha_0^2+\alpha_1^2)+(\beta_0-\beta_1)^2$ with multiplicity 2.
The corresponding normalized eigenvectors can be chosen as $e_0=\id_{\bb C^2}$,
\begin{equation}\label{eq:eigen1}
  e_\lambda
  =
    \begin{pmatrix}
      \tfrac{\alpha_0}{\alpha_1}    & 0 \\
    0 & -\tfrac{\alpha_1}{\alpha_0}
  \end{pmatrix},
    \quad
  e_{\mu^+}
  =\sqrt{1+\tfrac{\alpha_0^2}{\alpha_1^2}}
    \begin{pmatrix}
    0 & 1 \\
    0 & 0
  \end{pmatrix},
     \quad
  e_{\mu^-}
  =\sqrt{1+\tfrac{\alpha_1^2}{\alpha_0^2}}
    \begin{pmatrix}
    0 & 0 \\
    1 & 0
  \end{pmatrix}.
\end{equation}
In particular, the spectral gap of $-L_0$ is $\lambda_1=\lambda\wedge\mu$.
Moreover, again by direct computation, $\eta=\max\{\|e_\lambda\|_A, \|e_{\mu^+}\|_A, \|e_{\mu^-}\|_A\}=(\alpha_0/\alpha_1)\vee (\alpha_1/\alpha_0)$.

The decomposition in Section~\ref{sec:mai} is then achieved as follows.
Set  $\mathcal{I}_0=\mathcal{I}=\bb Z^d\times\{1,2,3\}$ and let $\imath\colon \mathcal{I}_0\to\mathcal{I}$ be the identity map.
Denoting by $a_x,b_x$ the elements in $\mathcal{A}_{\{x\}}$ corresponding to $a,b$ via the identification of $A$ with $\mathcal{A}_{\{x\}}$, we then set $\ell^0_{(x,j)}=\sigma_{x,j}a_x$, $j=1,2$, $\ell^0_{(x,3)}=b_x$, and $\ell^1_{\alpha}=\ell_{\alpha}-\ell_{\alpha}^0$, $\alpha\in\mathcal{I}$.
By construction
\begin{equation*}
  \mathcal{L}
  =\sum_{\alpha\in\mathcal{I}_0}\big(
{\ell^0_{\alpha}}^*[\, \cdot \, ,\ell^0_{\alpha}] +
[{\ell_{\alpha}^0}^*,\,\cdot\, ]\ell_{\alpha}^0
\big)+
\sum_{\alpha\in\mathcal{I}}L^1_\alpha
\end{equation*}
in which $L_\alpha^1$ is given by the right-hand side of \eqref{L1a} with $k_\alpha=0$.

Recalling \eqref{eq:C0} and \eqref{a.6}, by few trite computations we get $C_0=4\eta\sum_{j=1}^3\|\ell_{0,j}\|$ and
\begin{equation*}
  M\leq
    72\eta^2(\eta^2+1)(2R+1)^d\big(2(|\alpha_0|\vee |\alpha_1|)\Delta_{1,2}+2(|\beta_0|\vee |\beta_1|)\Delta_3+\Delta^2_{1,2}+\Delta^2_3\big).
\end{equation*}
In particular, the QFS generated by $\mathcal{L}$ is ergodic if $\Delta_{1,2}$  and $\Delta_3$ are small enough compared to the spectral gap $\lambda_1$.

\subsection{Quantum spin systems: conjugation with classical spin systems}
\label{sec:qss1}

According to the terminology in \cite[Ch.III]{Li}, a (classical) spin system is a Markov semigroup $(\mathcal{P}_t^\mathrm{cl})_{t\geq0}$ on the commutative $C^*$-algebra $C(\Omega)$ of the continuous $\bb C$-valued functions on $\Omega:=\{-1,1\}^{\bb Z^d}$ whose generator acts on local functions by
\begin{equation*}
  (\mathcal{L}_\mathrm{cl}f_\mathrm{cl})(\sigma)
  =\sum_{x\in\bb Z^d}c_x(\sigma)[f_\mathrm{cl}(\sigma^x)-f_\mathrm{cl}(\sigma)],
  \qquad\sigma\in\Omega,
\end{equation*}
where $c_x\colon\Omega\to[0,\infty)$ is the flip rate and $\sigma^x$ is the configuration in which the spin $\sigma$ is flipped at site $x$.
Many popular models like the stochastic Ising model, the contact process and the voter model are examples of spin systems.
Provided the flip rates $c_x$ satisfy suitable conditions, the semigroup generated by $\mathcal{L}_\mathrm{cl}$ is ergodic.
On the other hand, for particular choices of the flip rates ergodicity fails, i.e.\ the stationary state is not unique.
We refer to \cite{Li} for the details of both situations.

When the jump rates $c_x$ have finite range, we next construct a QFS $(\mathcal{P}_t)_{t\geq0}$ whose action on a commutative subalgebra $\mathcal{D}$ of $\mathcal{A}$ is conjugate to the one of $(\mathcal{P}_t^\mathrm{cl})_{t\geq0}$, thus providing, for non-ergodic $(\mathcal{P}_t^\mathrm{cl})_{t\geq0}$, examples of non-ergodic QFS.
The corresponding Lindblad generator has the form considered in the previous section, see in particular \eqref{eq:qss}.
Let $\mathcal{D}\subset\mathcal{A}$ be the commutative $C^*$-subalgebra of $\mathcal{A}$ generated by $\sigma_{x,3}$, $x\in\bb Z^d$.
Consider the $C^*$-algebra isomorphism $\imath\colon C(\Omega)\to\mathcal{D}$ defined by $\imath(f_\mathrm{cl})=f_\mathrm{cl}(\{\sigma_{x,3}\}_{x\in\bb Z^d})$.
By choosing $\ell_{x,j}=\imath(\sqrt{c_x})\sigma_{x,j}$, $j=1,2$ and $\ell_{x,3}\in\mathcal{D}$, a direct computation shows that $\mathcal{L}\circ\imath=\imath\circ\mathcal{L}_\mathrm{cl}$.
Hence the QFS $(\mathcal{P}_t)_{t\geq0}$ generated by $\mathcal{L}$ leaves invariant $\mathcal{D}$ and its action on $\mathcal{D}$ is conjugated to the one of $(\mathcal{P}_t^\mathrm{cl})_{t\geq0}$.

\subsection{$\boldsymbol{XYZ}$-model with site dissipation}
\label{sec:app}

In this section we consider a Heisenberg perturbation induced by the $XYZ$-Hamiltonian of non-interacting dissipative spins.
We show that, if the interaction parameters are small enough the resulting evolution is ergodic.

As in the previous sections, let $H=\bb C^2$, $A=\mathcal{B}(H)$, and $\sigma_j$, $j=1,2,3$ be the Pauli matrices.
The one-qubit unperturbed Lindblad generator is
\begin{equation*}
  L_0
  =\frac{1}{4}\sum_{j=1}^2\big( \sigma_j [\,\cdot\,,\sigma_j]+[\sigma_j, \,\cdot\,]\sigma_j\big)
\end{equation*}
which is self-adjoint with respect to the GNS inner product induced by the state $\rho=(1/2)\id_H$.
Note that this generator is a particular case of the one introduced in \eqref{eq:Lin1} when $a=\id_H$ and $b=0$.
In particular, the eigenvalues of $-L_0$ are $\lambda_0=0$, $\lambda_1=\lambda_2=2$, $\lambda_3=4$, with corresponding normalized eigenvectors given by \eqref{eq:eigen1} with $\alpha_0=\alpha_1=1/2$ and $\beta_0=\beta_1=0$
\begin{equation*}
  e_0=\id_H,\quad
  e_1=\sqrt{2}
  \begin{pmatrix}
    0 & 1 \\
    0 & 0
  \end{pmatrix}
  ,\quad
  e_2=\sqrt{2}
  \begin{pmatrix}
    0 & 0 \\
    1 & 0
  \end{pmatrix},
  \quad
  e_3=\sigma_3.
\end{equation*}
Hence $\eta=\max_{j\in\{1,2,3\}}\|e_j\|_{H\to H}=\sqrt{2}$.
As in Section~\ref{sec:udq}, we denote by $\mathcal{L}_0$ the Lindblad generator in which each qubit evolves independently according to $L_0$, see equation \eqref{mcL0}.

Given $J_j\in\bb R$, $j=1,2,3$, let $k_{X,j}=J_j\sigma_{x,j}\sigma_{y,j}$ if $X=\{x,y\}$ with $|x-y|=1$, and $k_{X,j}=0$ otherwise.
Here we understand, as in the previous section, $\sigma_{x,j}=\sigma_j\otimes\id_{\mathcal{H}_{\{x\}^c}}$, $j=1,2,3$.
The $XYZ$ model is then defined in terms of the informal Hamiltonian $\mathcal{K}=\sum_{X\in\Pfin}\sum_{j=1}^3k_{X,j}$.
Accordingly, we set $\mathcal{L}=\mathcal{L}_0+\mathcal{L}_1$ where $\mathcal{L}_1$ is the informal Heisenberg operator $\mathcal{L}_1=i[\mathcal{K},\,\cdot\,]$.
To fit this case in the general framework of Section~\ref{sec:diq}, set $\mathcal{I}=\Pfin\times \{1,2,3\}$, let $\chi\colon\mathcal{I}\to\Pfin$ be the projection on the first coordinate, $k_\alpha$ as defined above, $\ell_{(\{x\},j)}=(1/2)\sigma_{x,j}$, $x\in\bb Z^d$, $j=1,2$, and $\ell_\alpha=0$ otherwise.
Note that the family $\{k_\alpha,\ell_\alpha\}_{\alpha\in\mathcal{I}}$ is translation covariant and has range $R=1$.

As by definition $\mathcal{L}=\mathcal{L}_0+\mathcal{L}_1$, to apply the results of Section~\ref{sec:mai} it is enough to set $\mathcal{I}_0=\bb Z^d\times\{1,2\}$ with $\imath\colon\mathcal{I}_0\to\mathcal{I}$ given by $\imath(x,j)=(\{x\},j)$.
By direct computation $\|k_{(\{x,y\},j)}\|\leq|J_j|$ for $|x-y|=1$ and therefore the constant in \eqref{eq:C0} satisfies $C_0\leq 2\sqrt{2}(1+2d|J|)$, where $|J|=\sum_{j=1}^3|J_j|$.
Recalling \eqref{a.6}, few trite computations yield $M\leq 96\sqrt{2}d|J|$.
In particular, by Theorem~\ref{thm2}, the QFS generated by $\mathcal{L}$ is ergodic whenever $|J|<(48\sqrt{2}d)^{-1}$.

\section{
Interacting fermions
} 
\label{s:2-b}

In this section we consider quantum lattice systems described in terms of ferm\-io\-nic operators satisfying the \emph{canonical anticommutation relations} (CAR).
The unperturbed dynamics is given by the Fermi Ornstein-Uhlenbeck semigroup, while the interaction will be expressed as a superposition of local generators.

\subsection{Canonical anticommutation relations and fermionic $\boldsymbol{C^*}$-algebra}
\label{sec:car}

Referring e.g.\ to \cite{Mey} for the abstract setting of Clifford algebras,
we next introduce a family of operators satisfying the CAR in a concrete representation that describes fermions 
on the whole lattice $\bb Z^d$.

Let $\mathcal{H}$ be the Hilbert space with complete orthonormal system $\{e_X\}_{X\in\Pfin}$, where we recall that $\Pfin$ denotes the family of the finite subsets of $\bb Z^d$.
For $\Lambda\subset\bb Z^d$ we also consider the subspace
$\mc H_\Lambda$ spanned by $\{e_X\}_{X\subset\Lambda}$.
If $\Lambda$ is finite, then $\mc H_{\Lambda}$ has dimension $2^{|\Lambda|}$.
Clearly, $\mc H^0=\bigcup_{\Lambda\in\Pfin}\mc H_\Lambda$ is dense in $\mc H$.
Moreover $\mc H_{\Lambda_1\sqcup\Lambda_2}\simeq\mc H_{\Lambda_1}\otimes\mc H_{\Lambda_2}$, where $\sqcup$ denotes the union of  disjoint sets.

Denote by $\mc B(\mc H)$ the $C^*$-algebra of bounded operators on $\mc H$ and by $\|\cdot\|$ the corresponding norm.
As in the previous sections $\ind\in\mc B(\mc H)$ is the identity.
Fix a total order $\leq$ on $\bb Z^d$ and  for $x\in\bb Z^d$ define $a_x\in\mc B(\mc H)$ by
\begin{equation*}
  a_xe_X
  =\ind_X(x)(-1)^{|\{y\in X\colon y<x\}|}e_{X\setminus\{x\}}
 ,\qquad
 X\in\Pfin
\end{equation*}
where $\ind_X$ denotes the indicator function of the set $X$.
Accordingly,
\begin{equation*}
  a_x^*e_X
  =\ind_{\bb Z^d\setminus X}(x)(-1)^{|\{y\in X\colon y<x\}|}e_{X\cup\{x\}}
 ,\qquad
 X\in\Pfin.
\end{equation*}
By direct computations, the family $\{a_x,a_x^*\}_{x\in\bb Z^d}$ satisfies the CAR, i.e.\
\begin{equation}\label{b-car}
  \{a_x,a_y\}=\{a^*_x,a^*_y\}=0, \qquad
  \{a_x,a^*_y\}=\delta_{x,y} \ind
\end{equation}
where $\{a,b\}=ab+ba$ is the anticommutator of $a$ and $b$.
Let $n_x=a_x^*a_x$, $x\in\bb Z^d$, be the fermionic number operators.
These operators are pairwise commuting, self-adjoint, satisfy $n_x^2=n_x$, and act on $\mathcal{H}$ as $n_xe_X=\ind_X(x)e_X$.

For $\Lambda\in\Pfin$, let $\mc A_\Lambda\subset\mc B(\mc H)$ be the subalgebra generated by $\{a_x,a^*_x\}_{x\in\Lambda}$ and  set $\mc A^0:=\bigcup_{\Lambda\in\Pfin}\mc A_\Lambda$.
Noticing that $\mc A^0$ is a $*$-subalgebra of $\mc B(\mc H)$, we finally let $\mc A$ be the norm closure of $\mc A^0$.
In particular, $\mc A$ is a unital $C^*$-subalgebra of $\mc B(\mc H)$.
In fact, it is a proper subalgebra as, for example, the translation operators $\tau_x$, $x\in\bb Z^d$, 
defined by $\tau_x(e_X)=e_{X+x}$, lie in $\mc B(\mc H)$ but not in $\mc A$.
As in the previous section, we denote by $\mathcal{S}$ the collection of states on $\mathcal{A}$.

Note that the Hilbert space $\mathcal{H}$ constructed above can be identified with the one presented in Section~\ref{sec:udq} when $H=\bb C^2$.
However, under this  identification, the $C^*$-algebra $\mathcal{A}$ defined here does not coincide with the $C^*$-algebra $\mathcal{A}$ defined in Section~\ref{sec:udq}, since the fermionic operators $a_x$, $a_x^*$ are not local.

In order to define the dynamics, we next introduce another family of operators $\{v_x,v_x^*\}_{x\in\bb Z^d}$ satisfying the CAR. They will have the property that $v_x$ and $a_y$ commute  for $x\neq y$.
To this end, let $w$ be the self-adjoint and unitary element of $\mc B(\mc H)$ given by
$$
w e_X=(-1)^{|X|}e_X
,\qquad X\in\Pfin.
$$
The operator $w$ is usually referred to as the main automorphism or sign operator \cite{Mey}.
Observe that it does not belong to $\mathcal{A}$ and 
\begin{equation}\label{b-wa}
wa_x=-a_xw,\qquad wa_x^*=-a_x^*w,\qquad x\in\bb Z^d.
\end{equation}
Hence, letting $\Ad_w(f)=wfw$, $f\in\mc B(\mc H)$, for each $\Lambda\in\Pfin$, the subalgebra $\mc A_\Lambda$ is left invariant by $\Ad_w$.
In particular, $\Ad_w$ defines an outer automorphism of $\mc A$.

Define also the operators in $w\mc A\subset\mc B(\mc H)$
\begin{equation}\label{b-wv}
v_x=wa_x\,\,,\qquad
v_x^*=a_x^*w
\,,\qquad x\in\bb Z^d\,.
\end{equation}
Readily, also the family $\{v_x,v_x^*\}_{x\in\bb Z^d}$ satisfies the CAR.
Moreover, by direct computations,
\begin{equation}
  \label{b-wv1}
  [v_x,a_y]=[v_x^*,a_y^*]=0,\qquad
    [v_x,a_y^*]=-[v_x^*,a_y]=\delta_{x,y}w
    \,,\qquad x,y\in\bb Z^d\,.
\end{equation}
For $\Lambda\in\Pfin$ we let $\mc V_{\Lambda}$ be the finite-dimensional
$C^*$-subalgebra of $\mc B(\mc H)$ generated by $\{v_x,\,v^*_x\}_{x\in\Lambda}$.
As for the algebra $\mc A$, set 
$\mc V^0:=\bigcup_{\Lambda\in\Pfin}\mc V_\Lambda\subset\mc B(\mc H)$, 
and let $\mc V$ be its norm closure.
As follows from \eqref{b-wv1}, for disjoint $X,Y\in\Pfin$  the algebras $\mathcal{V}_{X}$ and $\mathcal{A}_Y$ commute, that is
\begin{equation}
  \label{eq:5}
  [u,f]=0,\quad \text{for any $u\in \mathcal{V}_{X}$, $f\in\mathcal{A}_Y$ with $X\cap Y=\emptyset$.}
\end{equation}

For $\Lambda\in\Pfin$, the algebras $\mc A_{\Lambda}$ and $\mc V_{\Lambda}$ are $\bb Z/2\bb Z$-graded, by letting $\deg(a_x)=\deg(a_x^*)=1$ and $\deg(v_x)=\deg(v_x^*)=1$, respectively.
Hence, we have the decompositions
$\mc A_{\Lambda}=\mc A_{\Lambda,0}\oplus\mc A_{\Lambda,1}$
and
$\mc V_{\Lambda}=\mc V_{\Lambda,0}\oplus\mc V_{\Lambda,1}$,
where $\mc A_{\Lambda,p}$ and $\mc V_{\Lambda,p}$
are the subspaces of parity $p\in\bb Z/2\bb Z$, respectively.
In view of \eqref{b-wa},
$\mc V_{\Lambda,0}=\mc A_{\Lambda,0}$ 
and $\mc V_{\Lambda,1}=w\mc A_{\Lambda,1}$.

Following \cite{CM1,G} we introduce a gradient structure induced by the fermionic operators.
Let $\partial_x,\,\bar{\partial}_x$, $x\in\bb Z^d$, be the bounded operators
on $\mc A$ defined by 
\begin{equation}\label{b-partialx}
\partial_x
:=
w[v_x,\,\cdot\,]
\,\,,\qquad\qquad
\bar{\partial}_x
:=
-w[v^*_x,\,\cdot\,]
\,.
\end{equation}
In view of \eqref{b-wa} and \eqref{b-wv}, if $f\in\mc A_{\{x\}^\mathrm{c}}$ then $\partial_xf=\bar{\partial}_xf=0$. 
Moreover, $\partial_x$ and $\bar \partial_x$ are skew derivations in the sense that for each $f,g\in\mathcal{A}$
\begin{equation*}
  \partial_x(fg)
  = (\partial_xf)g+\Ad_w(f) \partial_xg,
  \qquad
    \bar\partial_x(fg)
  = (\bar\partial_xf)g+\Ad_w(f) \bar\partial_xg.
\end{equation*}
Let also $\check\partial_x,\,\check{\bar\partial}_x$, $x\in\bb Z^d$, be the bounded operators
on $\mc V$ defined by
\begin{equation}\label{b-partialxc}
\check\partial_x
:=
w[a_x,\,\cdot\,]
\,\,,\qquad\qquad
\check{\bar\partial}_x
:=
-w[a^*_x,\,\cdot\,]
\,.
\end{equation}

\subsection{Dynamics}
\label{sec:ud}

As unperturbed dynamics we consider the Fermi Ornstein-Uh\-len\-be\-ck semigroup introduced 
in \cite{CM1,G}, that describe the dissipative evolution of free fermions.

For $h\in\bb R$  let $\pi_0\in\mathcal{S}$ be the product state corresponding to free fermions with external field $h$.
More precisely, for $\Lambda\in\Pfin$ and $f\in\mc A_{\Lambda}$, we set
$$
\pi_0(f)
=
\frac{\tr_{\mc H_\Lambda}\big(
e^{h\sum_{x\in\Lambda}n_x}f
\big)}
{(1+e^h)^{|\Lambda|}}
$$
which, by the density of $\mathcal{A}^0$ in $\mathcal{A}$, uniquely defines $\pi_0$.

We then define the unperturbed generator on $\mathcal{A}$ by
\begin{equation}
          \label{b-lho-f-0}
      \mc{L}_0 = \sum_{x\in\bb Z^d} L_x^0
    \end{equation}
    where
    \begin{equation}
        \label{b-lho-f}
        L_x^0=
        e^{h/2}
  \big( [v_x, \,\cdot\, ]v_x^* + v_x [ \,\cdot\, ,v_x^*] \big)
  +e^{-h/2} \big( [v_x^*, \,\cdot\, ]v_x + v_x^* [\,\cdot\,,v_x] \big) .
    \end{equation}
By direct computation, 
$L_x^0\mc A\subset\mc A$, and $\mathcal{L}_0$ is symmetric with respect to GNS inner  product induced by $\pi_0$.

In this section, we define the seminorm $\vertiii{\cdot}$ on $\mc A^0$ by
\begin{equation}\label{b-3norm}
\vertiii{ f} 
:=
\sum_{x\in\bb Z^d}(\|\partial_x f\|+\|\bar\partial_x f\|).
\end{equation}
Observe that $\vertiii{ f}=0$ if and only if $f$ is a scalar multiple of $\ind$.
Let also $\mc A^1$ be the closure of $\mc A^0$ with respect to the norm $\|\cdot\|+\vertiii{\cdot}$.
Clearly, $\mc A^0\subset\mc A^1\subset\mc A$.

We consider perturbations of the generator $\mc{L}_0$, both of conservative and dissipative type.
The local Hamiltonians will be assumed, as natural from a physical viewpoint, to be even functions of the fermionic operators, while the jump operators associated to the dissipative perturbation, belonging to $\mathcal{V}$, can be both even and odd but we will require them to have a definite parity.
More precisely, fix a countable set $\mathcal{I}$ and functions $p\colon\mathcal{I}\to\bb Z/2\bb Z$ and $\chi\colon\mathcal{I}\to\Pfin$ with finite fibers $\chi^{-1}(X)$, $X\in\Pfin$.
Denote also $\mathcal{I}_0=p^{-1}(0)$ and $\mathcal{I}_1=p^{-1}(1)$.
Fix then a collection $\{k_\alpha,\ell_\alpha\}_{\alpha\in \mathcal{I}}$ such that $k_\alpha=k_\alpha^*\in\mc V_{\chi(\alpha),0}=\mc A_{\chi(\alpha),0}$ for $\alpha\in\mathcal{I}_0$, $k_\alpha=0$ for $\alpha\in\mathcal{I}_1$, and $\ell_{\alpha}\in\mc V_{\chi(\alpha),p(\alpha)}$ for every $\alpha\in\mathcal{I}$.
We then set
\begin{equation}\label{KL2F}
\mc L_1
=
\sum_{\alpha\in \mathcal{I}}L^1_\alpha,
\quad\text{where}\quad
  L^1_\alpha 
  =i[k_\alpha,\,\cdot\,] 
  +[\ell_\alpha^*, \,\cdot\, ]\ell_\alpha + \ell_\alpha^* [\,\cdot\,,\ell_\alpha]
  \,.
\end{equation}
The above parity assumptions guarantee that $L^1_\alpha\mathcal{A}\subset\mathcal{A}$, $\alpha\in \mathcal{I}$.
We will show in Lemma \ref{KUF} and Theorem \ref{thm1F} that, under suitable conditions on the family $\{k_\alpha,\ell_\alpha\}_{\alpha\in\mathcal{I}}$,
the right-hand side of \eqref{KL2F} is well defined on  $\mc A^1$,
and the graph norm closure of $\mathcal{L}=\mathcal{L}_0+\mathcal{L}_1$ generates a QFS on $\mc A$.

As in Section~\ref{s:2}, the family $\big\{k_\alpha,\ell_\alpha\big\}_{\alpha\in\mathcal{I}}$,
has \emph{finite range} if there exists $R\in[0,\infty)$ such that $k_\alpha=\ell_\alpha=0$ whenever $\diam(\chi(\alpha))>R$.
The family $\big\{k_\alpha,\ell_\alpha\big\}_{\alpha\in\mathcal{I}}$ is \emph{tran\-slation covariant} if there exists an action of the abelian
group $\bb Z^d$ on $\mathcal{I}$, denoted by
$(x,\alpha)\mapsto x+\alpha$, satisfying
$\chi(x+\alpha)=x+\chi(\alpha)$, such that
$\Ad_{\tau_x}(k_\alpha) = k_{x+\alpha}$ and
$\Ad_{\tau_x}(\ell_\alpha) = \ell_{x+\alpha}$.

\subsection{Main results}

As we next state, under suitable assumptions, the operator $\mathcal{L}=\mathcal{L}_0+\mathcal{L}_1$ is well defined on $\mathcal{A}^1$.

\begin{lemma}\label{KUF}
  If
  \begin{equation}
    \label{eq:C0F}
    C_0
    :=2\varch(h/2)+4 \sup_{x\in\bb Z^d} \sum_{\alpha\colon\!\chi(\alpha) \ni
      x}(\|k_\alpha\|+2\|\ell_{\alpha}\|^2) < +\infty,
  \end{equation}
  then for each $f\in\mathcal{A}^1$ the series defining $\mathcal{L}f$
  converges in $\mathcal{A}$ and
  $\|\mathcal{L}f\|\leq C_0\vertiii{f}$.
\end{lemma}

As in the case of qudits, we first show the existence of the dynamics of interacting fermions.
For $\Lambda\in\Pfin$ we denote by $\mathcal{L}_\Lambda$ the bounded Lindblad generator on $\mathcal{A}$ defined by
\begin{equation*}
  \mathcal{L}_\Lambda
  :=\sum_{x\in\Lambda}L^0_x+\sum_{\alpha\in\mathcal{I}\colon \chi(\alpha)\subset\Lambda}L_\alpha^1
\end{equation*}
and by $(\mathcal{P}_t^\Lambda)_{t\geq0}$ the corresponding QFS.

\begin{theorem}\label{thm1F}
Assume \eqref{eq:C0F} and consider $\mc L=\mc L_0+\mc L_1$
as an operator on $\mc A$ with domain $\mc A^1$.
If
\begin{equation}\label{eqMF}
  \begin{split}
    M:=&4\sup_{y\in\bb Z^d}\sum_{x\in\bb Z^d}\sum_{\alpha\colon\chi(\alpha)\ni y}
    \Big(\|\partial_xk_\alpha\|+\|\bar\partial_xk_\alpha\|\\
    &+ 2 \|\ell_\alpha\|\big( \|\check\partial_x\ell_\alpha\|+ \|\check\partial_x\ell_\alpha^*\| + \|\check{\bar\partial}_x\ell_\alpha\|+ \|\check{\bar\partial}_x\ell_\alpha^*\|\big)\Big)<\infty,
  \end{split}
\end{equation}
then
\begin{enumerate}[(i)]
\item
the graph norm closure $\bar{\mc L}$ of $\mc L$ 
generates a QFS $(\mc P_t)_{t\geq0}$ on $\mc A$;
\item
$\mc A^1$ is a core for $\bar{\mc L}$;
\item for each $t\geq0$ the operator $\mathcal{P}_t$ is the strong limit of $\mathcal{P}_t^\Lambda$ as $\Lambda\uparrow\bb  Z^d$;
\item
for any $f\in\mc A^1$ and $t\geq0$ we have
$$
\vertiii{\mc P_tf}
\leq
e^{(M-2\varch(h/2))t}\vertiii{f}
\,;
$$
\item
  the QFS $(\mathcal{P}_t)_{t\geq0}$ has at least one stationary state.
\end{enumerate}
\end{theorem}

We refer to \cite{Re} for an alternative approach to the construction of the infinite volume dynamics.
The perturbative criterion for the ergodicity of interacting fermions is then stated as follows.

\begin{theorem}\label{thm2F}
Assume \eqref{eq:C0F} and $M<2\varch(h/2)$.
Then
\begin{enumerate}[(i)]
\item
the QFS $(\mc P_t)_{t\geq0}$ has a unique stationary state $\pi$;
\item
for any $f\in\mc A^1$ and $t\geq0$
$$
\|\mc P_tf-\pi(f)\ind\|
\leq
\frac{C_0}{2\varch(h/2)-M} e^{-(2\varch(h/2)-M)t}\vertiii{f}
\,;
$$
\item
if furthermore $\big\{k_\alpha,\ell_\alpha\big\}_{\alpha\in\mathcal{I}}$ has finite range,
there exist $C,\zeta>0$ such that for any $\Lambda_1,\Lambda_2\subset\bb Z^d$ and any $f_1\in\mc A_{\Lambda_1}$, $f_2\in\mc A_{\Lambda_2}$, 
$$
|\pi(f_1f_2)-\pi(f_1)\pi(f_2)|
\leq
C e^{-\zeta\,\mathrm{dist}(\Lambda_1,\Lambda_2)}\big(\|f_1\|+\vertiii{f_1}\big)\big(\|f_2\|+\vertiii{f_2}\big)
\,.
$$
\end{enumerate}
\end{theorem}

As $\gap(-L_x^0)=2\varch(h/2)$, the above criterion corresponds to the one formulated in Theorem~\ref{thm2} for qudits.
As in Section~\ref{s:2}, in the translation covariant case, Theorem~\ref{thm2F} implies the exponential decay of specific quantum one-Wasserstein distance between $\mu\mathcal{P}_t$ and $\pi$.
We refer to Theorem~\ref{thm3} for the precise statement.

\subsection{Bounds on the commutators}

The proofs of Theorems \ref{thm1F} and \ref{thm2F} are accomplished by the arguments presented in Section~\ref{s:2-c}.
The relevant ingredients are: (i) the intertwining relationships, as discussed in \cite{CM1} and stated in Lemma~\ref{b-lem24F} below, between the derivatives $\partial_x$, $\bar\partial_x$ and the unperturbed generator $\mathcal{L}_0$; (ii) some quantitative bounds on the commutators between the derivatives $\partial_x$, $\bar\partial_x$ and the perturbed generator $\mathcal{L}_1$ that are here derived anew.

We start with the following lemma which provides the fermionic counterpart to \eqref{eq:spdec}.

\begin{lemma}
  Let $E_{x}$ be the normalized partial trace on $\mathcal{H}_{\{x\}}$ and set
  \begin{equation*}
    D_x=a^*_x\partial_x=v_x^*[v_x,\,\cdot\,],
    \qquad
    \bar D_x=a_x\bar\partial_x=v_x[v_x^*,\,\cdot\,],
  \end{equation*}
  that are regarded as bounded operators on $\mathcal{A}$.
  Then for each $x\in\bb Z^d$ and $f\in\mathcal{A}$
  \begin{equation}\label{eq:dgrande}
    f=
    E_{x}f+D_xf+\bar D_xf-\frac{1}{2}\big(D_x\bar D_x+\bar D_x D_x\big)f.
  \end{equation}
\end{lemma}

\begin{proof}
  By linearity and density it suffices to show \eqref{eq:dgrande} when $f$ is a monomial;
  namely,  when $f=\otimes_yf_y$ with $f_y\in \mathcal{A}_{\{y\}}$ and the product runs over finitely many sites.
  In view of  \eqref{b-wv1}, for such $f$
  \begin{equation*}
    D_xf=\big(\otimes_{y<x}f_y\big) (D_xf_x)\big(\otimes_{y>x}f_y\big),
    \quad
    \bar D_xf=\big(\otimes_{y<x}f_y\big) (\bar D_xf_x)\big(\otimes_{y>x}f_y\big).
  \end{equation*}
  On the other hand, by direct computations,
    \begin{displaymath}
    \begin{array}{lllll}
     &D_x\ind=0, &D_xa_x=0, &D_xa_x^*=a_x^*, &D_xn_x=n_x, \\
    &\bar D_x\ind=0, &\bar D_xa_x=a_x, &\bar D_xa_x^*=0, &\bar D_xn_x=n_x-\ind.
  \end{array}
  \end{displaymath}
  As the linear span of $\{\ind,a_x,a_x^*,n_x\}$ is $\mathcal{A}_{\{x\}}$, the statement follows by linearity.  
\end{proof}

\begin{proof}[Proof of Lemma \ref{KUF}]
Since $\|w\|=\|a_x\|=\|a_x^*\|=1$ and $w^2=\ind$, we have
\begin{equation}\label{b-vxf}
\|[v_x,f]\|=\|\partial_xf\|
\,,\quad
\|[v_x^*,f]\|=\|\bar{\partial}_xf\|
\,,\quad
\|[n_x,f]\|\leq\|\partial_xf\|+\|\bar{\partial}_xf\|
\,.
\end{equation}
Hence, from \eqref{b-lho-f} we get
\begin{equation}
\label{b-eq:stima1}
\|L_x^0f \|
\leq 2\varch(h/2)\big(\|\partial_xf\|+\|\bar\partial_xf\|\big).
\end{equation}

As follows from a direct computation, the statement is achieved by the following estimate.
For each  $X\in\Pfin$, $u\in\mathcal{V}_X$, and $f\in\mathcal{A}^1$
\begin{equation}\label{eq:uf}
  \|[u,f]\|
  \leq 4\|u\|\sum_{x\in X}(\|\partial_x f\|+\|\bar\partial_x f\|).
\end{equation}

To prove the bound \eqref{eq:uf}, let $X=\{x_1,\dots,x_m\}$ and introduce the operator $F_x\colon\mathcal{A}\to\mathcal{A}$ defined by $F_x=D_x+\bar D_x-(1/2)(D_x\bar D_x+\bar D_x D_x)$.
By recursively using \eqref{eq:dgrande} we deduce, cf.\ \eqref{eq:Fx},
\begin{equation}\label{eq:f=}
  f=\Big(\prod_{j=1}^mE_{x_j}\Big)f+\sum_{j=1}^m\Big(\prod_{h=1}^{j-1}E_{x_h}\Big)F_{x_j}f.  
\end{equation}
Since $\|E_x\|_{\mathcal{A}\to\mathcal{A}}\leq1$, $\|D_x\|_{\mathcal{A}\to\mathcal{A}}\leq2$, $\|\bar D_x\|_{\mathcal{A}\to\mathcal{A}}\leq2$, $\|D_xf\|\leq \|\partial_xf\|$, and $\|\bar D_xf\|\leq \|\bar\partial_xf\|$, we deduce 
\begin{equation}\label{eq:fleq}
  \Big\|\Big(\prod_{h=1}^{j-1}E_{x_h}\Big)F_{x_j}f\Big\|
  \leq 2\big(\|\partial_{x_j}f\|+\|\bar\partial_{x_j}f\|\big),
  \qquad j=1,\dots, m.
\end{equation}
Recalling \eqref{eq:5}, the claim \eqref{eq:uf} follows by noticing that $\big(\prod_{j=1}^mE_{x_j}\big)f$ belongs to $\mathcal{A}_{X^\mathrm{c}}$.
\end{proof}

We next recall the intertwining relationship  between the unperturbed generator $\mathcal{L}_0$ and the gradient structure introduced in \eqref{b-partialx}.

\begin{lemma}\label{b-lem24F}
	For each $f\in{\mc A}^1$ and $x\in{\bb Z}^d$
	\begin{align*}
		\partial_x {\mc L}_0 f-{\mc L}_0\partial_x f&= -2\varch(h/2)\partial_xf\\
				\bar\partial_x {\mc L}_0 u-{\mc L}_0\bar\partial_x f&=-2\varch(h/2)\bar\partial_xf.
	\end{align*}
\end{lemma}
\begin{proof}
Both identities are obtained by a straightforward computation, see also \cite[\S6.2]{CM1}.
\end{proof}

\begin{lemma}$~$\label{lem:iii}
  Fix $x\in\bb Z^d$ and $X\in\Pfin$.
  \begin{enumerate}[(i)]
  \item For each $u\in\mathcal{V}_{X,0}$ and $f\in\mathcal{A}^1$
    \begin{align*}
      \|\partial_x[u,f]-[u,\partial_x f]\|
      \leq 4 \|\partial_x u\|\sum_{y\in X}\big(\|\partial_yf\|+\|\bar\partial_yf\|\big),\\
       \|\bar\partial_x[u,f]-[u,\bar\partial_x f]\|
      \leq 4 \|\bar\partial_x u\|\sum_{y\in X}\big(\|\partial_yf\|+\|\bar\partial_yf\|\big).
    \end{align*}
    \item  For each $j=0,1$, $u\in\mathcal{V}_{X,j}$, and $f\in\mathcal{A}^1$
      \begin{align*}
       & \big\|\partial_x(u^*[f,u]+[u^*,f]u)-u^*[\partial_xf,u]-[u^*\partial_xf]u\big\|\\
        &\qquad\leq 8\|u\|\big(\|\check\partial_xu^*\|+\|\check\partial_xu\|\big)\sum_{y\in X}\big(\|\partial_yf\|+\|\bar\partial_yf\|\big),\\
               & \big\|\bar\partial_x(u^*[f,u]+[u^*,f]u)-u^*[\bar\partial_xf,u]-[u^*\bar\partial_xf]u\big\|\\
        &\qquad\leq 8\|u\|\big(\|\check{\bar\partial}_xu^*\|+\|\check{\bar\partial}_xu\|\big)\sum_{y\in X}\big(\|\partial_yf\|+\|\bar\partial_yf\|\big).
      \end{align*}
    \end{enumerate}
\end{lemma}

\begin{proof}
  We use the notation introduced in the proof of Lemma~\ref{KUF}.
  For (i), we prove only the first bound.
  Since $u\in\mathcal{V}_{X,0}$, the Jacobi identity yields
  \begin{equation*}
    \partial_x[u,f]-[u,\partial_x f]
    =w[[v_x,u],f]
    =w \sum_{j=1}^m\Big[[v_x,u], \Big(\prod_{h=1}^{j-1}E_{x_h}\Big)F_{x_j}f\Big],
  \end{equation*}
  where we used \eqref{eq:f=} and \eqref{eq:5} in the second step.
  The statement now follows from \eqref{eq:fleq}.

  Regarding (ii), we prove  again  only the first bound.
  By direct computation, for $j=0,1$, $u\in\mathcal{V}_j$, $f\in\mathcal{A}$,
  \begin{align*}
    &\partial_x(u^*[f,u])-u^*[\partial_xf,u]
    =w\big((\check\partial_xu^*)[f,u]+(-1)^ju^*[f,\check\partial_xu]\big)\\
    &\partial_x([u^*,f]u)-[u^*,\partial_xf]u
    =w\big([\check\partial_xu^*,f]u+(-1)^j[u^*,f](\check\partial_xu)\big).
  \end{align*}
  The statement now follows by \eqref{eq:uf}.
\end{proof}

As we next state, the estimates in Lemma~\ref{lem:iii} provide the required bounds on the commutator between the gradient structure introduced in \eqref{b-partialx} and perturbed generator $\mathcal{L}_1$.

\begin{lemma}\label{b-lem245}
  For $x,y\in\bb Z^d$ set
  \begin{align*}
    &\theta_{x,y}
      :=
      4\sum_{\alpha\colon y\in\chi(\alpha)}
      \big(\|\partial_xk_\alpha\|+2  \|\ell_\alpha\|( \|\check\partial_x\ell_\alpha\|+ \|\check\partial_x\ell_\alpha^*\|)\big)\\
    &\tilde\theta_{x,y}
      :=
      4\sum_{\alpha\colon y\in\chi(\alpha)}
      \big(\|\bar\partial_xk_\alpha\|+2  \|\ell_\alpha\|( \|\check{\bar\partial}_x\ell_\alpha\|+ \|\check{\bar\partial}_x\ell_\alpha^*\|)\big).
\end{align*}
Then for each $x\in{\bb Z}^d$ and $f\in{\mc A}^1$
	\begin{align*}
	& \|\partial_x {\mc L}_1 f-{\mc L}_1\partial_x f\|
	\leq
	\sum_{y\in\bb Z^d} \theta_{x,y}
	( \|\partial_y f\| + \|\bar{\partial}_y f\| )
	\,,\\
	& \|\bar{\partial}_x {\mc L}_1 f-{\mc L}_1\bar{\partial}_x f\|
	\leq
	\sum_{y\in\bb Z^d} \tilde\theta_{x,y}
	( \|\partial_y f\| + \|\bar{\partial}_y f\| ).
	\end{align*}
\end{lemma}

\begin{proof}
  The result is a direct consequence of Lemma~\ref{lem:iii}. 
\end{proof}

Recalling \eqref{eqMF}, note that
\begin{equation*}
  M=\sup_{y\in\bb Z^d}\sum_{x\in\bb Z^d}(\theta_{x,y}+\tilde \theta_{x,y}).
\end{equation*}
Lemmata \ref{KUF}, \ref{b-lem24F}, and \ref{b-lem245} provide the ingredients to achieve the proofs of Theorems \ref{thm1F} and \ref{thm2F} by the arguments presented in Section~\ref{s:2-c}.

\subsection{Nearest neighbor interacting fermions with  site dissipation} 
\label{s:xxf}

Consider the informal Hamiltonian given by 
\begin{equation*}
  \mathcal{K} =J\sum_{\{x,y\}\colon |x-y|=1}\big(a_{x}^*a_y+a_y^*a_{x}\big),
\end{equation*}
for some $J\in\bb R$.
Letting $\mathcal{L}_0$ be the Lindblad generator of the Fermi Ornstein-Uhlenbeck QFS introduced in \eqref{b-lho-f-0},  
consider the QFS with informal generator
\begin{equation*}
  \mathcal{L}=\mathcal{L}_0+i[\mathcal{K},\cdot\,].
\end{equation*}
It fits the setting introduced in Section~\ref{sec:ud} with the translation covariant conservative interaction parametrized by $\mathcal{I}=\big\{\{x,y\}\colon x,y\in\bb Z^d, |x-y|=1\big\}\subset\Pfin$ and given by
\begin{equation*}
  k_{\{x,y\}}=J \big(a_{x}^*a_y+a_y^*a_{x}\big),
\end{equation*}
understanding that $\chi\colon\mathcal{I}\to\Pfin$ is the inclusion map.
In particular, the family $\{k_\alpha\}_{\alpha\in\mathcal{I}}$ has range 1.
By direct computation, $\|k_{\{x,y\}}\|=|J|$ and
\begin{equation*}
  \theta_{x,y}=\tilde\theta_{x,y}
  =
  \begin{cases}
    8|J|d & \text{if $x=y$,}\\
    4|J| & \text{if $|x-y|=1$,}\\
    0 & \text{otherwise.}
  \end{cases}
\end{equation*}
Correspondingly, $M=32d|J|$, so that the assumptions of Theorem~\ref{thm2F} are met when $|J|<\varch(h/2)/(16d)$. 

\section*{Declarations}
L.~Bertini has been co-funded by the European Union (ERC CoG KiLiM, project number 101125162).

A.~De Sole is a member of the GNSAGA INdAM group and he has been supported by the national PRIN grant 2022S8SSW2.

Data sharing is not applicable to this article as no new data were created or analyzed in this study.

The authors have no competing interests to declare that are relevant to the content of this article.

\section*{Acknowledgments}
We are grateful to Eric A.~Carlen for a fruitful discussion.
It is a pleasure to thank D.~Trevisan for explaining us the basics of quantum information theory.
We finally thank an anonymous referee for pointing us out some relevant references and helpful comments.

\end{document}